\documentclass[11pt]{article}
\label{thm:intro-confused-collector-main}
\usepackage[margin=1in]{geometry}

\usepackage[dvipsnames]{xcolor}
\usepackage[utf8]{inputenc}
\usepackage{mathtools,amsfonts,amsthm,mathrsfs,amssymb}
    \usepackage{stmaryrd}
\usepackage{textcomp}
\mathtoolsset{showonlyrefs}
\usepackage{microtype}
\usepackage{todonotes}
\usepackage{enumerate}
\usepackage{enumitem}
\usepackage[indent]{parskip}
\usepackage[hidelinks,
						bookmarksnumbered,
						bookmarksopen=true,
						colorlinks=true,
						linkcolor=blue!70!black, 
						citecolor=green!50!black, 
						urlcolor=red!50!black]{hyperref}

\usepackage{xspace}
\newcommand*\ie{i.\kern.1em e.\ }
\newcommand*\eg{e.\kern.1em g.\ }

\makeatletter
\newtheorem*{rep@theorem}{\rep@title}
\newcommand{\newreptheorem}[2]{
\newenvironment{rep#1}[1]{
 \def\rep@title{#2 \ref{##1}}
 \begin{rep@theorem}}
 {\end{rep@theorem}}}
\makeatother

\newtheorem{thm}{Theorem}[section]
\newreptheorem{thm}{Theorem}
\newtheorem{lemma}[thm]{Lemma}

\newtheorem{proposition}[thm]{Proposition}
\newtheorem{corollary}[thm]{Corollary}
\newtheorem{conj}[thm]{Conjecture}

\theoremstyle{definition}
\newtheorem{definition}{Definition}
\newtheorem*{definition*}{Definition}
\newtheorem{rem}[thm]{Remark}
\newtheorem{example}[thm]{Example}
\newtheorem{qn}{Question}



\newtheorem*{acknowledgement}{Acknowledgments}

\newcommand{\zo}{\{0,1\}}
\newcommand{\define}{\vcentcolon=}

\newcommand{\G}{\mathcal{G}}

\newcommand{\R}{\mathbb{R}}

\newcommand{\ignore}[1]{}

\newcommand{\sk}{\mathsf{sk}}
\newcommand{\eps}{\varepsilon}

\newcommand{\wcol}{\mathrm{wcol}}

\renewcommand{\le}{\leqslant}
\renewcommand{\leq}{\leqslant}
\renewcommand{\ge}{\geqslant}
\renewcommand{\geq}{\geqslant}

\def\F {{\mathcal F}}
\def\G {{\mathcal G}}

\def\A {{\mathcal A}}
\def\P {{\mathcal P}}
\def\X {{\mathcal X}}

\def\fG{{\mathfrak G}}

\def\bN{{\mathbb N}}

\def\sk {\mathsf{sk}}

\newcommand{\dist}{\mathsf{dist}}
\newcommand{\ind}[1]{\mathbf{1}[ #1 ]}

\title{Sketching Distances in Monotone Graph Classes\thanks{A
    preliminary version of this work appeared in the proceedings of
    the conference \emph{Approximation, Randomization, and
      Combinatorial Optimization. Algorithms and Techniques}
    (APPROX/RANDOM 2022) \cite{EHK22}.}}

\author{%
Louis Esperet\thanks{Partially supported by the French ANR Projects
  GATO (ANR-16-CE40-0009-01), GrR (ANR-18-CE40-0032), TWIN-WIDTH
  (ANR-21-CE48-0014-01) and by LabEx
  PERSYVAL-lab (ANR-11-LABX-0025).}\\ Laboratoire G-SCOP, CNRS, France \\ \texttt{louis.esperet@grenoble-inp.fr}
\and
Nathaniel Harms\thanks{This work was partly funded by an NSERC Canada Graduate Scholarship.} \\
EPFL, Switzerland \\
\texttt{nathaniel.harms@epfl.ch}
\and
Andrey Kupavskii \\ Laboratoire G-SCOP, CNRS, France,\\ Moscow Institute of Physics, and Technology, Russia and \\Huawei R\&D Moscow, Russia \\
\texttt{kupavskii@ya.ru}
}

\begin{document}

\maketitle

\begin{abstract}
We study the two-player communication problem of determining whether two vertices $x,y$ are nearby
in a graph $G$, with the goal of determining the graph structures that allow the problem to be
solved with a \emph{constant-cost} randomized protocol. Equivalently, we consider the problem of
assigning constant-size random labels (\emph{sketches}) to the vertices of a graph, which allow adjacency, exact
distance thresholds, or approximate distance thresholds to be computed with high probability from
the labels.

Our main results are that, for monotone classes of graphs: constant-size adjacency sketches exist if
and only if the class has bounded arboricity; constant-size sketches for exact distance thresholds
exist if and only if the class has bounded expansion; constant-size approximate distance threshold
(ADT) sketches imply that the class has bounded expansion; any class of constant expansion (i.e.~any
proper minor closed class) has constant-size ADT sketches; and a class may have arbitrarily small
expansion without admitting constant-size ADT sketches.
\end{abstract}

\thispagestyle{empty}
\setcounter{page}{0}
\newpage

\section{Introduction}\label{sec:intro}

We are interested in understanding the power of \emph{constant-cost}, public-coin, randomized
communication, which has been the subject of a number of recent works
\cite{Har20,HHH21b,HWZ21,CHZZ22,HHM22,HHP+22}. Here,
two players compute a function of their inputs (with success probability at least, say, $2/3$) using
a randomized protocol whose \emph{cost} (number of bits communicated) is independent of the size of
the inputs. Examples include the \textsc{Equality} problem, where two players must decide if they
had the same input. In practice, such protocols are used, for example, as checksums (using, say, the
SHA256 hash function).

One natural type of problem is if two players have vertices $x$ and $y$ in a graph $G$, and would
like to decide if their vertices are nearby. We say $G$ belongs to a class of graphs $\F$, and we
ask which classes $\F$ allow the problem to be solved by a protocol whose cost is independent of the
size of $G \in \F$. An equivalent question (see \eg \cite{Har20, HWZ21}) is to ask for a random
assignment of \emph{constant-size labels} to each vertex of $G$ (where the number of bits in each
label does not depend on the number of vertices of $G$), which we call \emph{sketches}, in
such a way that one can determine whether two vertices $x$ and $y$ are nearby from their sketches. There are a
number of ways one might define ``nearby''\!\!. To give some examples, consider the case where
$\F$ is the class of hypercube graphs (whose vertices are binary strings $\zo^n$, with an edge
between $x$ and $y$ if they differ on a single bit):
\begin{enumerate}[itemsep=0em,leftmargin=1.75em]
\item Adjacency in the hypercube can be computed (with probability at least $2/3$) from sketches of
constant size (which follows from the Hamming distance communication protocol~\cite{HSZZ06});
\item Distinguishing between $\dist(x,y) \leq r$ and $\dist(x,y) > r$ can be done with
sketches of size depending only on $r$ (which also follows from the Hamming distance protocol);
\item Distinguishing between $\dist(x,y) \leq r$ and $\dist(x,y) > \alpha r$ (for constant
$\alpha > 1$) can be done with sketches of size independent of $r$ and $n$ \cite{KOR00}.
\end{enumerate}
\noindent
We call these \emph{adjacency sketches}, \emph{small-distance sketches}, and \emph{approximate
distance threshold (ADT) sketches}, respectively (see Section~\ref{section:results} for formal
definitions). We would like to know which classes of graphs, other than the hypercubes, admit
similarly efficient sketches.  Sketches for deciding $\dist(x,y) \leq r$ vs.~$\dist(x,y) > \alpha r$
in general metric spaces (especially normed spaces \eg \cite{Ind06,AKR18,KN19}) are well-studied,
and characterizing the metrics which admit this type of sketch is a well-known open problem
\cite{SS02,AK08,sublinear25,Raz17}, but little is known about the natural case of path-distance
metrics in graphs. Recent work \cite{HWZ21} asked which \emph{hereditary} classes of graphs admit
constant-size adjacency sketches, motivated by a connection between communication complexity and
graph labelling schemes. \cite{HWZ21} also gives some examples of constant-size small-distance
sketches, including for planar graphs, answering a question of \cite{Har20}.

We study the relationships between these three types of sketches for the important special case of
\emph{monotone} classes of graphs. A class of graphs is a set of (labelled\footnote{Standard
terminology is that a \emph{labelled} $n$-vertex graph is one with vertex set $[n]$; not to be
confused with \emph{informative labelling schemes}.}) graphs closed under isomorphism. It is
\emph{hereditary} if it is closed under taking induced subgraphs, and \emph{monotone} if it is
closed under taking subgraphs. Monotone graph classes are ubiquitous: typical examples include
minor-closed classes, graphs avoiding some subgraph $H$, or graphs with bounded chromatic number.

In this paper, we completely determine the monotone graph classes which admit constant-size
adjacency sketches and constant-size (\ie independent of the number of vertices) small-distance
sketches, and show that constant-size (\ie independent of the number of vertices and the parameter
$r$) ADT sketches imply the existence of constant-size small-distance sketches. Our main tool is a
new connection between communication complexity and \emph{sparsity theory} of graphs \cite{NO12}. We
show that the monotone classes which admit constant-size adjacency sketches are exactly the classes
with bounded arboricity, and the monotone classes which admit constant-size small-distance sketches
are exactly the classes with bounded \emph{expansion}\footnote{We mean bounded expansion in the
sense of sparsity theory \cite{NO12}, which is distinct from expansion in the context of expander
graphs.}.  Monotone classes which admit constant-size ADT sketches must also have bounded expansion,
and any class with constant expansion (\ie any proper minor-closed class) has a constant-size ADT
sketch, but on the other hand a class can have expansion growing arbitrarily slowly and yet does
\emph{not} admit a constant-size ADT sketch. We describe these results in more detail below.

\subsection{Motivation and Related Work}
Labelling schemes and sketches are important primitives for distributed computing, streaming,
communication, data structures for approximate nearest neighbors, and even classical algorithms (see
e.g.~\cite{KNR92,GP03,Spin03,Pel05,EIX22}, and \cite{AMS99,Ind06,AK08,Raz17,AKR18} and references
therein).  As such, a great deal of research has been done on finding other spaces having nice
sketching and labelling properties.

One direction of research investigates the metric spaces which admit approximate distance threshold
(ADT) sketches, of the third type described above, as defined in \cite{SS02}. This is a well-known
open problem in sublinear algorithms (see e.g.~\cite{AK08,sublinear25,Raz17}).  Here, $n$ points $X
\subseteq \X$ in a metric space $(\X, \dist)$, should be assigned random sketches $\sk: X \to \zo^*$
such that $\dist(x,y) \leq r$ or $\dist(x,y) \geq \alpha r$ can be determined (with probability at
least $2/3$) from $\sk(x)$ and $\sk(y)$.  The goal is to obtain sketches whose size depends only on
$\alpha$. This problem is fairly well-understood when the metric is a norm: there is a constant-size
sketch for the $\ell_p$ (quasi-)norm, for any $0 < p \leq 2$ \cite{Ind06}, so any metric that can be
embedded into such an $\ell_p$ is sketchable; conversely, sketching a norm is equivalent to
embedding it into $\ell_{1-\eps}$ \cite{AKR18}. Outside of norms, the problem is less
well-understood: there are sketchable metrics that are not embeddable into $\ell_{1-\eps}$
\cite{KN19}. 

Another direction of research investigates the classes $\F$ of graphs that admit (deterministic)
labelling schemes for various functions, generally called \emph{informative labelling schemes}
\cite{Pel05}. The most well-studied labelling schemes are for adjacency, introduced in
\cite{KNR92,Mul89}. The main open problem is to identify the hereditary classes of graphs that admit
adjacency labelling schemes of size $O(\log n)$. A solution was suggested in \cite{KNR92} and later
conjectured in \cite{Spin03}, but recently refuted in a breakthrough of \cite{HH21}, leaving the
problem wide open. \emph{Randomized} adjacency labelling (\ie \emph{adjacency sketching}) was
studied in \cite{FK09,Har20,HWZ21}. It was observed in \cite{Har20,HWZ21} that a constant-size
sketch implies an $O(\log n)$ labelling scheme, as desired in the above open problem, and it was
further observed in \cite{HWZ21} that the set of hereditary graph classes which admit constant-size
adjacency sketches is equivalent to the set of Boolean-valued communication problems that admit
constant-cost public-coin protocols, whose structure is unknown \cite{HHH21b}.  This raises the
following question, which was the main motivation of \cite{HWZ21}:

\begin{qn}
\label{question:adjacency}
Which hereditary classes of graphs admit constant-size adjacency sketches?
\end{qn}

\noindent
Perhaps the next most commonly studied graph labelling problem is \emph{distance labelling}
\cite{GPPR04}, where the goal is to compute $\dist(x,y)$ from the labels (see
e.g.~\cite{ADKP16,AGHP16a,FGNW17,GU21}).  Intermediate between distance and adjacency labelling is
the decision version of distance labelling: for given $r$, decide whether $\dist(x,y) \leq r$ from
the labels. We call this \emph{small-distance labelling}, following the terminology of
\cite{ABR05,GL07}.  For $r=1$, this coincides with adjacency labelling. The natural generalization
of constant-size adjacency sketches is to ask for small-distance sketches whose size depends only on
$r$; it was shown in \cite{Har20} that such sketches exist for trees, and in \cite{HWZ21} that they
exist for any Cartesian product graphs and any \emph{stable}\footnote{See \cite{HWZ21} for a
discussion of stability, which is not necessary for the current paper.} class of bounded
twin-width (including, for example, planar graphs or any proper minor-closed
class; see \cite{GPT21}).

\begin{qn}
\label{question:small-distance}
Which hereditary classes of graphs admit small-distance sketches whose size depends only on $r$?
\end{qn}

\noindent
It is common to weaken distance labelling to \emph{approximate} distance labelling \cite{GKK+01},
where the goal is to approximate $\dist(x,y)$ up to a constant factor (see
e.g.~\cite{Thor04,ACG12,AGHP16b}). The decision version is to distinguish, for a
given $r$, between $\dist(x,y) \leq r$ and $\dist(x,y) > \alpha r$; we will call this problem
\emph{$\alpha$-approximate distance threshold (ADT)} labelling and sketching. This is a similar
formulation as the distance sketching problem mentioned above, with the $n$ points from the metric
space $\X$ being replaced with a size $n$ graph from a class $\F$. Despite significant interest in
distance sketching and labelling, the only prior work explicitly relating the two, or studying
\emph{randomized} ADT labelling, appears to be the unpublished manuscript \cite{AK08} (although
there is extensive literature on the related problem of embedding graph metrics into normed spaces
\cite[Chapter 15]{Mat13}; embedding planar graphs into $\ell_1$ with constant
distortion is a major open problem \cite{GNRS04}). This raises the following question, which is a
special case of the open problem of identifying sketchable metrics:

\begin{qn}
\label{question:adt}
Which classes of graphs admit constant-size ADT sketches?
\end{qn}

\noindent
It holds by definition (see definitions below) that a small-distance sketchable class $\F$ is
adjacency sketchable, but the relationships between other types of sketching are otherwise unclear,
\emph{a priori}. It seems reasonable to suspect that these three types of sketching require similar
conditions on the graph class $\F$; so we ask:
\noindent
\begin{qn}
\label{question:relationship}
What is the relationship between adjacency, small-distance, and ADT
sketching?
\end{qn}

Finally, the adjacency and small-distance sketches we obtain in this paper turn out to be
\emph{equality-based}, meaning that the associated randomized communinication protocols can be
simulated by \emph{deterministic} communication protocols that have access to an oracle which
computes \textsc{Equality} (we explain this more carefully below). Communication with the
\textsc{Equality} oracle has recently become a topic of interest in communication complexity; see
\cite{GPW18,CLV19,HHH21b,HWZ21,AY22,PSS23}.

\subsection{Our Results}
\label{section:results}
In this paper, we resolve Questions \ref{question:adjacency}, \ref{question:small-distance}, and
\ref{question:relationship} for \emph{monotone} classes of graphs, and make progress towards
Question \ref{question:adt}. The sketches we obtain usually do not assume that the
classes under consideration are monotone, but our lower bounds
crucially rely on this assumption. We first formally define the main three types of \emph{sketchability} that
we are concerned with. We will generalize these definitions in Section~\ref{section:sketching
definitions}.  For a graph class $\F$, we say:
\begin{enumerate}[itemsep=0pt]
\item $\F$ admits an \emph{adjacency sketch of size $s(n)$} if there is a function $D : \zo^* \times
\zo^* \to \zo$ such that $\forall G \in \F$ with size $n$, there is a random function $\sk : V(G)
\to \zo^{s(n)}$ satisfying
\[
  \forall x,y \in V(G) : \qquad
    \Pr\left[ D(\sk(x), \sk(y)) = 1 \iff x, y \text{ are adjacent} \right] \geq 2/3 \,.
\]
$\F$ is \emph{adjacency sketchable} if it admits an adjacency sketch of constant size.
\item $\F$ admits a \emph{small-distance sketch of size $s(n,r)$} if for every $r \in \bN$ there is
a function $D_r : \zo^* \times \zo^* \to \zo$ such that $\forall G \in \F$ with size $n$, there is a
random function $\sk : V(G) \to \zo^{s(n,r)}$ satisfying
\[
  \forall x,y \in V(G) : \qquad
    \Pr\left[ D_r(\sk(x), \sk(y)) = 1 \iff \dist_G(x,y) \leq r \right] \geq 2/3 \,.
\]
$\F$ is \emph{small-distance sketchable} if it admits a small-distance sketch of size
independent of $n$.
\item For constant $\alpha > 1$, $\F$ admits an \emph{$\alpha$-ADT sketch of size $s(n)$} if for
every $r \in \bN$ there is a function $D_r : \zo^* \times \zo^* \to \zo$ such that $\forall G \in
\F$ with size $n$, there is a random function $\sk : V(G) \to \zo^{s(n)}$ satisfying
\begin{align*}
  \forall x,y \in V(G) : \qquad
    \dist(x,y) \leq r &\implies \Pr\left[ D_r(\sk(x), \sk(y)) = 1 \right] \geq 2/3 \\
    \dist(x,y) > \alpha r &\implies \Pr\left[ D_r(\sk(x), \sk(y)) = 0 \right] \geq 2/3 \,.
\end{align*}
For a constant $\alpha > 1$, we say that $\F$ is \emph{$\alpha$-ADT sketchable} if $\F$ admits an
$\alpha$-ADT sketch with size independent of $n$. $\F$ is
\emph{ADT sketchable} if there is a constant $\alpha > 1$ such that
$\F$ is $\alpha$-ADT sketchable.
We discuss some nuances of ADT sketch
size in Section~\ref{section:adt}.
\end{enumerate}
\noindent
Our results imply the following hierarchy, which answers Question~\ref{question:relationship} for
monotone classes of graphs.  Let $\mathsf{ADJ}$ be the adjacency sketchable monotone graph classes,
$\mathsf{SD}$ the small-distance sketchable monotone graph classes, and $\mathsf{ADT}$ the ADT
sketchable monotone graph classes. Then
\[
\mathsf{ADT} \subsetneq \mathsf{SD} \subsetneq \mathsf{ADJ} \,.
\]
\noindent
That $\mathsf{SD} \subseteq \mathsf{ADJ}$ follows by definition, and $\mathsf{SD} \neq \mathsf{ADJ}$
is witnessed by the arboricity-2 graphs (as observed in \cite{Har20}). Our contribution to this
hierarchy is $\mathsf{ADT} \subsetneq \mathsf{SD}$ (which does not necessarily hold for non-monotone
classes, see Example~\ref{example:general adt}), a complete characterization of the sets
$\mathsf{SD}$ and $\mathsf{ADJ}$, and some results towards a characterization of $\mathsf{ADT}$.

\subsubsection{Adjacency Sketching}
We resolve Question~\ref{question:adjacency} for monotone classes by showing that they are adjacency
sketchable if and only if they have bounded arboricity. Moreover, we obtain
an adjacency sketch of size at most linear in the
arboricity, and show that this is best possible. The \emph{arboricity}
of a graph $G$ is the minimum integer $k$ such that the edges of $G$ 
can be partitioned into $k$ forests. A class $\F$ has arboricity at
most $k$ if all graphs $G \in \F$ have
arboricity at most $k$. If there exists some constant $k$ such that
$\F$ has arboricity at most $k$, we say
$\F$ has \emph{bounded arboricity}.
\begin{thm}
\label{thm:intro adjacency}
Any class of arboricity at most $k$ has an adjacency sketch of size
$O(k)$, and any monotone class containing a graph of arboricity $k$
requires adjacency sketches of size $\Omega(k)$. In particular, any
monotone class $\F$ is adjacency sketchable if and
only if $\F$ has bounded arboricity.
\end{thm}
\noindent
All proofs for adjacency sketching are in Section~\ref{section:adjacency}.  Using standard random
hashing and the adjacency labelling scheme of \cite{KNR92}, it is easy to see that any class of
bounded arboricity is adjacency sketchable; this was stated explicitly in \cite{Har20,HWZ21} (the
latter giving slightly improved sketch size) and the result for trees also appeared in \cite{FK09}.
We prove the converse for monotone classes (which does not hold for hereditary classes in general
\cite{HWZ21}).
We use a counting argument to show that for any graph $G$ of
arboricity $d$,  the class of all spanning
subgraphs of $G$ requires adjacency sketches of
size $\Omega(d)$; our approach is inspired by the recent proof of \cite{HHH21a,HHH21b} that refuted a
conjecture of \cite{HWZ21} regarding adjacency sketchable graph classes (see Conjecture
\ref{conj:hwz}), and we find that the subgraphs of the hypercube are a more natural counterexample
to the conjecture of \cite{HWZ21}.

The hashing-based sketch uses randomization only to compute \textsc{Equality} subproblems; \ie it
can be simulated by a constant-cost \emph{deterministic} communication protocol with access to a
unit-cost \textsc{Equality} oracle.  This type of sketch is called \emph{equality-based} in
\cite{HWZ21}.  Equality-based sketches imply some structural properties of the graph class, such as
the strong Erd{\H o}s-Hajnal property \cite{HHH21b}. Recent work has studied the power of the
\textsc{Equality} oracle and found that it does not capture the full power of randomization
\cite{CLV19,HHH21b,HWZ21}; in particular, the Boolean hypercubes (and any Cartesian product graphs)
are adjacency sketchable, but not with an equality-based sketch \cite{HHH21b,HWZ21}.  Our result
shows that \textsc{Equality} captures the power of randomization for sketching monotone classes of
graphs. In fact, it is only necessary to compute a \emph{disjunction} of equality checks, which we
think of as the simplest possible type of sketch. In terms of graph structure, a constant-size
equality-based sketch means that the graph can be written as the Boolean combination of a constant
number of equivalence graphs; if the Boolean combination is a disjunction, then the graph may simply
be covered by a constant number of equivalence graphs.

We remark that sketches (especially small-distance or ADT sketches) which compute a disjunction of
equality checks can be used to obtain \emph{locality-sensitive hashes}, a widely-used algorithmic
tool introduced in \cite{IM98}. Almost all of our positive results are of this type. See
Remark~\ref{remark:lsh}.

\subsubsection{Small-Distance Sketching}
We answer Question~\ref{question:small-distance} by proving that the monotone graph classes that are
small-distance sketchable are exactly those with \emph{bounded expansion} (as in \cite{NO12}); see 
Definition~\ref{def:boundedexp}. Informally, bounded expansion means that the edge density of a
graph increases only as a function of $r$ when contracting subgraphs of radius $r$ into a single
vertex. Many graph classes of theoretical and practical importance have bounded expansion, including
bounded-degree graphs, proper minor-closed graph classes, and graphs of bounded genus \cite{NO12},
along with many random graph models and real-world graphs \cite{DRR+14}.

To state our theorem, we briefly describe another type of sketch that generalizes small-distance
sketching, called \emph{first-order} sketching. A graph class $\F$ is \emph{first-order sketchable}
if any first-order (FO) formula $\phi(x,y)$ over the vertices and edge relation of the graph (with
two free variables whose domain is the set of vertices) is sketchable (see
Section~\ref{section:sketching definitions}). This type of sketch was introduced in \cite{HWZ21} and
generalizes small-distance sketching, along with (for example) testing whether vertices $x,y$ belong
to a subgraph isomorphic to some fixed graph $H$. We show that, for monotone graph classes,
first-order sketchability is equivalent to small-distance sketchability.  All proofs for
small-distance sketching are in Section~\ref{section:small-distance}.  

\begin{thm}
\label{thm:intro small-distance}
Let $\F$ be a monotone class of graphs. Then the following are equivalent:
\begin{enumerate}[itemsep=0pt]
\item $\F$ is small-distance sketchable;
\item $\F$ is first-order sketchable;
\item $\F$ has bounded expansion.
\end{enumerate}
\end{thm}
\noindent
The implications $(3) \implies (2) \implies (1)$ do not require monotonicity.  $(2) \implies (1)$
holds by definition.  The proof of $(3) \implies (2)$ is straightforward, but relies on a structural
result of \cite{GKN+20} whose proof is highly technical. We actually get the stronger result that
any class with \emph{structurally bounded expansion} (\ie any class that is a \emph{first-order
transduction} of a class with bounded expansion) is first-order sketchable, which improves the
results of \cite{HWZ21}. It was proved in \cite{HWZ21}, using structural results of \cite{GPT21},
that any stable class of bounded twin-width is first-order sketchable. A stable class has bounded
twin-width if and only if it is a transduction of a class of bounded sparse twin-width \cite{GPT21}.
Every class of bounded sparse twin-width has bounded expansion, but the converse does not hold
(e.g.~for cubic graphs)~\cite{BGK+21}, so our result generalizes the result of \cite{HWZ21}. It
essentially follows from using the structural results of \cite{GKN+20} instead of \cite{GPT21}.
Another interesting consequence of our result, in conjunction with concurrent work of \cite{HHP+22},
is that graph classes with structurally bounded expansion can be represented as a type of
constant-dimensional geometric intersection graphs; see Corollary~\ref{cor:sign-rank}.

Our proof of $(1) \implies (3)$ (Section~\ref{section:small-distance lower bounds}) requires our
proof of Theorem~\ref{thm:intro adjacency} and some results in sparsity theory \cite{KO04,NO12}.  We
actually prove a stronger statement: for any monotone class $\F$, the existence of a sketch for
deciding $\dist(x,y) \leq r$ vs. $\dist(x,y) > 5r-1$, with size depending only on $r$, implies
bounded expansion. Under a conjecture of Thomassen \cite{Tho83}, we can replace the constant 5 with
any arbitrarily large constant; see the remark after Conjecture~\ref{conj:thom}.  Note that, even
with a constant-factor gap between distance thresholds, this problem is distinct from ADT sketching,
since the small-distance sketch size is allowed to depend on $r$. If we could replace the constant 5
with any arbitrarily large constant, this would immediately imply $\mathsf{ADT} \subseteq
\mathsf{SD}$.

We also present a more direct proof of $(3) \implies (1)$, without going through first-order
sketching, that allows for quantitative results.  Going through first-order sketching (as was also
done in \cite{HWZ21}) proves the existence of a function $f(r)$ bounding the sketch size, without
giving it explicitly. We obtain explicit bounds in terms of the \emph{weak coloring number}
\cite{NO12}, written as $\wcol_r(\F)$ for any $r \in \bN$ (Definition~\ref{def:weak reachability}).
Using known bounds on the weak coloring number \cite{HOQRS17}, we obtain the following corollary. As
was the case for adjacency sketching, we observe that this proof (unlike the more general one for
first-order sketching) produces sketches that only use randomization to compute a disjunction of
\textsc{Equality} checks, establishing that this extremely simple type of sketch suffices for
monotone classes.

\begin{corollary}
Any graph class $\F$ with bounded $\wcol_r(\F)$ admits a small-distance sketch of size
$O(r+\wcol_r(\F)\log(\wcol_r(\F)))$. In particular, planar graphs admit a small-distance sketch of
size $O(r^3 \log r)$, and the class of $K_t$-minor-free graphs admits a small-distance sketch of
size $O(r^{t-1} \log r)$. Furthermore, planar graphs admit a small-distance labelling scheme of size
$O(r^3 \log n)$ and $K_t$-minor-free graphs admit a small-distance labelling scheme of size
$O(r^{t-1} \log n)$.
\end{corollary}

\noindent
Although it is not necessary for our main goal, it may be desirable for large values of $r$ to have
small-distance sketches with smaller dependence on $r$, at the expense of some dependence on the
graph size $n$. We present a proof that, for any fixed surface $\Sigma$, the class of graphs which
can be embedded\footnote{Here, we mean embedding in the sense of graph drawing, as opposed to metric
embedding.} in $\Sigma$ admits a small-distance labelling scheme of size $O(r \log^2 n)$; see
Theorem~\ref{thm:distr2}. This proof, using the layering technique of \cite{RS84,Epp00}, is due to
Gwena\"el Joret (personal communication); we thank him for allowing us to include it here.

\subsubsection{Approximate Distance Sketching}
\label{section:intro adt}
In light of Theorem~\ref{thm:intro small-distance}, a reasonable question is whether ADT sketching
for monotone classes is also determined by expansion. Our first result is that bounded expansion is
necessary. All proofs on approximate distance sketching are in Section~\ref{section:adt}.

\begin{thm}
If a monotone class $\F$ is ADT sketchable, then it has bounded expansion.
\end{thm}
\noindent
Combined with Theorem~\ref{thm:intro small-distance}, this proves $\mathsf{ADT} \subseteq
\mathsf{SD}$. Our proof uses a recent and fairly involved result in
extremal graph theory \cite{LM20}, along with the
theory of sparsity \cite{NO12}, to show that an $\alpha$-ADT sketch for a monotone class $\F$ of
unbounded expansion could be used to get a constant-size sketch for deciding $\dist(x,y) \leq 1$
vs.~$\dist(x,y) > \alpha$ in arbitrary graphs, which (as we show) is a contradiction.

We are then concerned with the converse. We show that the class of max-degree 3 graphs, which has
expansion exponential in $r$ \cite{NO08}, is not ADT sketchable. After proving this theorem, we
learned of an unpublished result \cite{AK08} which proves a $\Theta(\log(n) /\alpha)$ bound for
one-way communication of the $\alpha$-ADT problem on degree-3 expander graphs. This could be used in
place of our theorem to get the same qualitative (constant vs.~non-constant) results, but not the
quantitative bound: note that communication complexity cannot give sketching or labelling lower
bounds better than $\Theta(\log n)$.

\begin{thm}
\label{thm:intro bounded degree}
For any $\alpha > 1$, any $\alpha$-ADT sketch for the class of graphs with maximum degree 3 has size
at least $\Omega(n^{\tfrac{1}{4\alpha}-\eps})$, for any constant $\eps > 0$.
\end{thm}
\noindent
This establishes that $\mathsf{ADT} \neq \mathsf{SD}$ (and negatively answers open problem 2 of
\cite{AG06} about approximate distance labels for bounded-degree graphs, which \cite{AK08} does
not). But max-degree 3 graphs have exponential expansion.  Smaller bounds on the expansion are
associated with structural properties: for example, in monotone classes, polynomial expansion is
equivalent to the existence of strongly sublinear separators \cite{DN16}.  One may then wonder if
smaller bounds on the expansion suffice to guarantee ADT sketchability. We prove that this is not
the case for two natural examples: subgraphs of the 3-dimensional grid (with polynomial expansion
\cite{NO12}), and subgraphs of the 2-dimensional grid with crosses (with linear expansion
\cite{Dvo21}) are not ADT sketchable. For this we require our Theorem~\ref{thm:intro bounded
degree}.

\begin{proposition}
For the class of subgraphs of the 3-dimensional grid (the Cartesian product of 3 paths), and the
class of subgraphs of the 2-dimensional grid with crosses (the strong product of 2 paths), an $\alpha$-ADT sketch
requires size at least $n^{\Omega(1/\alpha)}$.
\end{proposition}

\noindent
We strengthen this result by showing that one can obtain monotone classes of graphs with expansion
that grows arbitrarily slowly, which are not ADT sketchable.

\begin{thm}
\label{thm:intro small expansion}
For any function $\rho$ tending to infinity, there exists a monotone class of expansion $r \mapsto
\rho(r)$ that is not ADT sketchable.  Moreover, for any $\eps>0$, there exists a monotone class
$\F$ of expansion $r\mapsto O(r^\eps)$, such that, if $\F$ admits an $\alpha$-ADT sketch of size
$s(n)$, then we must have $s(n) = n^{\Omega(1/\alpha)}$.
\end{thm}

\noindent
We conclude with a brief discussion of upper bounds for ADT sketching. A number of concepts have
been introduced in the literature that can be used to obtain ADT
sketches, including
sparse covers \cite{AP90} and padded decompositions \cite{KPR93}.
See Section \ref{sec:upperbounds} for definitions of these
concepts. We present in Theorem \ref{thm:SCdist} a construction of ADT sketches from sparse covers.

Using the sketches obtained from sparse covers, combined with results of \cite{Fil20} on sparse
covers (based on \cite{KPR93,FT03}), we obtain the following, which complements our
Theorem~\ref{thm:intro small expansion}; note that the graph classes with constant expansion are
exactly the proper minor-closed classes \cite{NO12}.
\begin{corollary}
For any $t \geq 4$, the class of $K_t$-minor-free graphs has a $O(2^t)$-ADT sketch of size $O(t^2 \log
t)$. The sketch is equality-based and has one-sided error. As a consequence, every monotone class of
constant expansion is ADT sketchable.
\end{corollary}
\noindent
It is also relatively straightforward (see Theorem \ref{thm:PDSdist}) to obtain
ADT sketches from padded decompositions, with an interesting difference. These sketches may not have
one-sided error and, unlike all other positive examples of sketches in this paper, they may not be
equality-based. On the other hand, they are extremely small. We can use constructions of padded
decompositions due to \cite{LS10,AGGNT19} to obtain the following remarkable corollary:

\begin{corollary}
For any $t \geq 4$, the class of $K_t$-minor-free graphs has an $O(t)$-ADT sketch of size 2. For $g
\geq 0$, the class of graphs embeddable on a surface of Euler genus $g$ has an $O(\log g)$-ADT
sketch of size 2.
\end{corollary}

\subsection{Discussion and Open Problems}

The main problem left open by this paper is Question~\ref{question:adt} for monotone
classes of graphs; we have shown that a constant bound on the expansion implies ADT sketchability,
while arbitrarily small non-constant bounds do not, but this does not rule out a monotone, ADT
sketchable class with non-constant expansion.

We have examples showing that ADT sketching does not imply small-distance
sketching, in general. But our examples are not even hereditary. Is there a hereditary class
that is ADT sketchable, but not small-distance or adjacency sketchable?

Our Theorem~\ref{thm:intro small-distance} shows that bounded expansion implies first-order
sketchability, and that for monotone classes the converse also holds. We showed more generally that
classes of \emph{structurally bounded expansion} are first-order sketchable. To extend our study of
sketchability beyond monotone classes, it would be interesting to investigate whether the converse
of this statement holds: does first-order sketchability of a \emph{hereditary} class imply
structurally bounded expansion?

In the preprint of this paper, we asked whether the class of subgraphs of hypercubes is a
counterexample to the Implicit Graph Conjecture (IGC):

\noindent
\textbf{Implicit Graph Conjecture} (see \eg \cite{KNR92,Spin03}): Does every hereditary graph class
$\F$ containing at most $2^{O(n \log n)}$ graphs on $n$-vertices, admit an adjacency labeling scheme
with labels of size $O(\log n)$?

\noindent
This long-standing conjecture was refuted in \cite{HH21} by a non-constructive argument, and
it would be interesting to find a more natural class that refutes the conjecture. One way to design
adjacency labels is to find a constant-size adjacency sketch \cite{Har20,HWZ21}, but our
Corollary~\ref{cor:hypercube counterexample} shows that this doesn't exist for the family of
subgraphs of hypercubes; furthermore, prior work (\eg \cite{CLR20}) had not found another way of
designing labels of size $O(\log n)$, suggesting that subgraphs of hypercubes might be
an explicit counterexample. But, follow-up work has now found efficient labels for this class
\cite{EHZ22}, leaving open the problem of finding an explicit counterexample to the IGC.  A related
question is whether we may characterize the monotone classes of graphs which admit adjacency
labeling schemes of size $O(\log n)$.

We have focused on determining whether there \emph{exists} a constant $\alpha$ such that a class is
$\alpha$-ADT sketchable. It is also of interest to obtain sketches for arbitrarily small $\alpha >
1$, with sketch size depending on $\alpha$. One strategy is to embed the graph isometrically into
$\ell_1$, but this is not always the best option. We obtained a $(1+\eps)$-ADT sketch for the
class of forests with size $O\left(\tfrac{1}{\eps}\log\tfrac{1}{\eps}\right)$, but this
result appeared earlier in \cite{AK08}; this sketch is more efficient than the one obtained by
embedding the trees isometrically in $\ell_1$. We remark that a class (monotone or not) that admits
a $(1+\eps)$-ADT sketch for $\eps < 1$ must also admit an adjacency sketch.

Finally, we point out an interesting conjecture of \cite{HHH21b}, that all constant-cost public-coin
communication problems contain a large monochromatic rectangle. In our terminology, using the
equivalence between constant-cost communication and adjacency sketching from \cite{HWZ21}, this
conjecture states that all adjacency sketchable graph classes have the strong Erd{\H o}s-Hajnal
property.

\section{Preliminaries}

\subsection{Notation}

Throughout the paper, $\log$ denotes the logarithm base 2, while $\ln$
denotes the natural logarithm.

\noindent
We will write $\ind{E}$ for the indicator variable for the event $E$, which takes value 1 if $E$ is
true.

\noindent
Given a graph $G$, the length of a path $P$ in $G$ is the number of edges of $P$.  Given two
vertices $x,y \in V(G)$, we define $\dist_{G}(x,y)$ to be the infimum of the length of a path
between $x$ and $y$ in $G$; we define $\dist_{G}(x,y)=\infty$ if there exists no path between $x$ and
$y$.  Notice that $(V(G),\dist_{G})$ is a metric space (with possibly infinite distances between pairs
of vertices if $G$ is disconnected).

\noindent
The \emph{girth} of a graph $G$ is defined as the size of a shortest
cycle in $G$ (if $G$ is acyclic, its girth is infinite).

\noindent For a class $\F$ of
graphs and an integer $n\ge 0$, we denote by $\F_n$ the family of
graphs from $\F$ with $n$ vertices.

\subsection{Distance and First-Order Sketching}
\label{section:sketching definitions}

We will require more general notions of sketching than those introduced above.  For a class $\F$ of
graphs, we will say that a sequence $\{ f_G \}_{G \in \F}$ of partial functions $f_G : V(G) \times
V(G) \to \{0,1,*\}$ is a \emph{partial function $f$ parameterized by graphs $G \in \F$}. We will
write $f$ to refer to this sequence.

For a graph class $\F$, we define an \emph{$f$-sketch} for $\F$ as a decoder $D : \zo^* \times \zo^*
\to \zo$, such that for every $G \in \F$ the following holds. There is a probability distribution
over functions $\sk : V(G) \to \zo^*$, such that for all $x,y \in V(G)$,
\[
  f_G(x,y) \neq * \implies \Pr[ D(\sk(x), \sk(y)) = f_G(x,y) ] \geq 2/3 \,.
\]
For a fixed  $f$-sketch for $\F$ and any graph $G \in \F$,  we call any probability distribution
over $\sk : V(G) \to \zo^*$ that satisfies the above condition an \emph{$f_G$-sketch} or an
\emph{$f$-sketch for $G$}.
\noindent
We define the \emph{size} of the sketch as
\[
\max_{G \in \F_n} \sup_{\sk} \max_{x \in V(G)} |\sk(x)| \,,
\]
\noindent
where the supremum is over the set of functions $\sk : V(G) \to \zo^*$ in the support of the
distribution defined for $G$, and $|\sk(x)|$ is the number of bits of $\sk(x)$. We will say that a
class $\F$ is \emph{$f$-sketchable} if there exists an $f$-sketch for $\F$ with size that does not
depend on the number of vertices $n$.

For a graph class $\F$, we also define an \emph{$f$-labelling scheme} for $\F$ similar to above,
except that for every $G \in \F$ there is a \emph{deterministic} function $\ell : V(G) \to \zo^*$
such that for all $x,y \in V(G)$,
\[
  f_G(x,y) \neq * \;\implies\; D(\ell(x), \ell(y)) = f_G(x,y) \,.
\]
\noindent
The following simple proposition (observed in \cite{Har20,HWZ21}) relates sketches to labelling
schemes:
\begin{proposition}
\label{prop:derandomization}
If $\F$ admits an $f$-sketch of size $s(n)$, then it admits an $f$-labelling scheme of size $O(s(n)
\log n)$.
\end{proposition}

\noindent
We now define certain important types of $f$-sketches.  Let $\F$ be a class of graphs. For any $r_1
\leq r_2$, a \emph{distance-$(r_1,r_2)$ sketch} for $\F$ is an $f$-sketch, as defined above, when
for any graph $G$ we define the function
\[
  f_G(x,y) = \begin{cases}
    1 &\text{ if } \dist_G(x,y) \leq r_1 \\
    0 &\text{ if } \dist_G(x,y) > r_2 \\
    * &\text{ otherwise.}
  \end{cases}
\]
\noindent
The size of such a sketch may depend on $r_1$, $r_2$, the number of vertices $n$, or other graph
parameters.

Recall the definitions of \emph{adjacency sketchable}, \emph{small-distance sketchable}, and
\emph{ADT sketchable}. It is clear that:
\begin{enumerate}
\item A class $\F$ is \emph{adjacency sketchable} if it is distance-$(1,1)$ sketchable;
\item A class $\F$ is \emph{small-distance sketchable} if for every $r \geq 1$ it is
distance-$(r,r)$ sketchable.
\item A class $\F$ is \emph{$\alpha$-ADT sketchable} if for every $r \geq 1$ it is
distance-$(r,\alpha r)$ sketchable, and furthermore the size of the sketch does not depend on $r$.
\end{enumerate}

\noindent
Following \cite{HWZ21}, we will also define \emph{FO-sketchable} classes, for which we require some
terminology (see \eg \cite{NOS22} for more on the following terminology). A \emph{relational
vocabulary} $\Sigma$ is a set of relation symbols, with each $R \in \Sigma$ having an \emph{arity}
$\mathrm{arity}(R) \in \bN\setminus \{0\}$. A $\Sigma$-structure $\A$ consists of a \emph{domain}
$A$, and for each relation symbol $R \in \Sigma$ an \emph{interpretation} $R^{\A} \subseteq
A^{\mathrm{arity}(R)}$, which is a relation. Fix a countably infinite set $X$ of \emph{variables}.
\emph{Atomic formulas of vocabulary $\Sigma$} are of the form
\begin{itemize}
\item $x=y$ for $x,y \in X$; or,
\item $R(x_1, \dotsc, x_r)$ for $x_1,\dotsc, x_r \in X$, $R \in \Sigma$ and $r=\mathrm{arity}(R)$,
which evaluates to true when $(x_1, \dotsc, x_r) \in R$.
\end{itemize}
\emph{First-order (FO) formulas} of vocabulary $\Sigma$ are inductively defined as either atomic
formulas, or a formula of the form $\neg \phi, \phi \wedge \psi, \phi \vee \psi$, or $\exists x .
\phi$ or $\forall x . \psi$, where $\phi$ and $\psi$ are each FO formulas. A \emph{free variable} of
a formula $\phi$ is one which is not bound by a quantifier. We will write $\phi(x_1, x_2, \dotsc,
x_k)$ to show that the free variables of $\phi$ are $x_1, \dotsc, x_k \in X$. For a value $u \in A$,
we write $\phi[ u / x ]$ for the formula obtained by substituting the constant $u$ for the free
variable $x$.

Let $\phi(x,y)$ be any formula with two free variables and relational vocabulary $\Sigma = \{ E',
R_1, \dotsc, R_k \}$ where $E'$ is symmetric of arity 2 and each $R_i$ is unary (\ie of arity 1). We
will say that a graph class $\F$ is \emph{$\phi$-sketchable} if it is $f$-sketchable for any $f$
chosen as follows. For any graph $G = (V,E)$, we choose any $\Sigma$-structure with domain $V$ where
$E$ is the interpretation of the symbol $E'$. Then set $f_G(u,v) = 1$ if and only if $\phi(u / x, v
/ y)$ evaluates to true.

We remark that for any graph $G$, there are many ways to choose a $\Sigma$-structure with domain $V$
with $E$ being the interpretation of $E'$. To be first-order sketchable, a class $\F$ must be
$f$-sketchable for \emph{every} such choice of functions $f_G$. A concrete example is that, for any
$r \in \bN$, we can choose the formula
\[
  \phi(x,y) = \exists u_1, u_2, \dotsc, u_{r-1} : (E'(x,u_1) \vee x=u_1) \wedge (E'(u_1,u_2) \vee
u_1=u_2) \wedge \dotsm \wedge (E'(u_r,y) \vee u_r=y) \,,
\]
which evaluates to true if and only if $\dist_G(x,y) \leq r$.

\subsection{Equality-Based Labelling Schemes and Sketches}\label{sec:dis}
An equality-based labelling scheme is one which assigns to each vertex a deterministic label,
comprising a data structure of size $s$ that holds $k$ ``equality codes''\!\!, which can be used
only for checking equality. These labelling schemes: 1) capture the
constant-cost randomized communication protocols that can be simulated by a constant-cost
\emph{deterministic} communication protocol with access to an \textsc{Equality} oracle (as studied
in e.g.~\cite{CLV19,BBM+20,HHH21b,HWZ21}); and 2) capture a common type of adjacency labels,
including those of \cite{KNR92} for bounded arboricity graphs (see \cite{HWZ21} for others).

One might formalize these schemes in a few ways; we slightly adapt the definition from
\cite{HWZ21}. This definition is intended to simplify notation rather than optimize label size,
since we care mainly about constant vs.~non-constant.

\begin{definition}[Equality-Based Labeling Scheme]
Let $\F$ be a class of graphs and let $f : \bN \times \bN \times \F \to \{0,1,*\}$ be a partial
function. An $(s,k)$-equality-based $f$-labeling scheme for $\F$ is an algorithm $D$, called a
\emph{decoder}, which satisfies the following. For every $G \in \F$ with vertex set $[n]$ and every
$x \in [n]$, there is a tuple of the form
\[
  \ell_G(x) = (p(x) \mid q(x)) \,,
\]
where $p(x) \in \zo^s$ is called the \emph{prefix}\footnote{It is natural but unnecessary to include
a prefix, since one may replace the $s$-bit prefixes $p(x)$ with $s+1$ single-bit equality codes
$(1, p(x))$, so that the prefixes $p(x), p(y)$ can be recovered by checking $1 = p(x)_i?$ and $1 =
p(y)_i?$ for each $i \in [s]$. This is convenient for lower bounds (\eg \cite{CLV19,HWZ21}).} and
$q(x) \in \bN^k$ is called the vector of \emph{equality codes}, such that, for all $x,y \in [n]$, on
inputs $\ell_G(x), \ell_G(y)$, the algorithm $D$ chooses a function $D_{p(x), p(y)} : \zo^{k \times
k} \to \zo$ and outputs
\[
  D_{p(x), p(y)}(Q_{x,y}) \,,
\]
where $Q_{x,y} \in \zo^{k \times k}$ is the matrix recording whether each pair of equality codes are
equal:
\[
  Q_{x,y}(i,j) \define \ind{q(x)_i = q(y)_j} \,.
\]
It is required that, for all $G \in \F$ and all $x,y \in [n]$, the output satisfies
\[
  D(\ell_G(x), \ell_G(y)) = D_{p(x), p(y)}(Q_{x,y}) = f(x,y,G) \,.
\]
We make the further distinction that an $(s,k)$-equality-based labelling scheme is
$(s,k)$-\emph{disjunctive} if for every $G \in \F$ and $x,y \in [n]$, $D_{p(x),p(y)}(Q_{x,y})$
simply outputs a disjunction of a subset of entries of $Q_{x,y}$.
\end{definition}
It will be convenient to introduce some alternate notation for constructing equality-based labeling
schemes. We define an $(s,t,k)$-equality-based $f$-labelling scheme for $\F$ similarly to an
$(s,k)$-equality-based labelling scheme, except for $G \in \F$ and $x \in [n]$, the labels are of
the form
\[
  \ell_G(x)
    = \left[
      (p_1(x) \mid \vec q_1(x)), (p_2(x) \mid \vec q_2(x)), \dotsc, (p_t(x) \mid \vec q_t(x))
      \right],
\]
where the vectors $p_i(x) \in \zo^*$ are called the \emph{prefixes}, the entries of the vectors
$\vec q_i(x) \in \bN^*$ are called \emph{equality codes}, and we must have $\sum_{i=1}^t |p_i(x)|
\leq s$ and $\sum_{i=1}^t |\vec q_i(x)| \leq k$ (where $|v|$ denotes the number of entries of $v$).
When an element $(p_i(x) \mid \vec q_i(x))$ in an equality-based label has $p_i(x)$ of size 0, we
will write $(- \mid \vec q_i(x))$; similarly, we write $(p_1(x) \mid - )$ when $\vec q_i(x)$ is
empty.  Given labels of this form, it is straightforward to obtain an $(s,k)$-equality based
labelling scheme by concatenating the prefixes and equality codes, along with separator symbols so
that the decoder can reconstruct each $(p_i \mid \vec q_i)$.

We emphasize that $k$ bounds the total number of equality codes associated with any vertex $x$, but
not necessarily the total number of bits needed to store these codes (see Example
\ref{example:disjunctive} below, where $k=2$ but storing the codes would require $2\log n$ bits per
vertex). 

\begin{example}
\label{example:disjunctive}
The adjacency labelling scheme of \cite{KNR92} for forests can be written as an equality-based
labelling scheme. For each $x$ in an $n$-vertex forest $G$ with arbitrarily rooted trees, which we
assume has vertex set $[n]$, we assign the label $\ell_G(x) = ( - \mid x,p(x) )$ where $p(x)$ is the
parent of $x$ if it has one, or 0 otherwise. Here the equality codes are $q(x) = (x,p(x)) \in
\bN^2$. The decoder simply outputs the disjunction of $p(x) = y$ or $p(y) = x$, so in fact this is a
$(0,2)$-disjunctive labeling scheme.
\end{example}

An equality-based labelling scheme is easily transformed into a standard deterministic labelling
scheme or a sketch.  The following simple proposition was observed in \cite{HWZ21}. We sketch the
proof for the sake of clarity. 
\begin{proposition}\label{prop:eqbased}
Let $\F$ be a class of graphs and $f : \bN \times \bN \times \F \to \{0,1,*\}$ be a partial
function. If there is an $(s,t,k)$-equality-based $f$-labelling scheme for $\F$ then there is an
$f$-sketch for $\F$ of size at most $O(s + t + k \log k)$. If the scheme is disjunctive, the sketch
has \emph{one-sided error}: when $f(x,y,G) = 1$, the sketch will produce the wrong output with
probability 0.
\end{proposition}
\begin{proof}[Proof sketch]
Choose a random function $\xi : \bN \to [w]$ for $w = 3k^2$. For any vertex $x$ of a graph $G$,
replace each vector $\vec q_i(x) = (q_{i,1}(x), \dotsc, q_{i,m}(x))$ with $(\xi(q_{i,1}(x)), \dotsc,
\xi(q_{i,m}(x)))$. We have replaced each of the (at most) $k$ equality codes $(\vec q_{i}(x))_j$
with $\xi((q_{i}(x))_j)$, using $k\log w = O(k \log k)$ bits in total. The sketch has size $O(s + t
+ k \log k)$ since we must include each $p_i(x)$ (using $s$ bits in total), the $O(k \log k)$ bits
for the equality codes, and $O(t)$ bits to encode the symbols $( \; | \; )$.

For two vertices $x,y$, write $Q_{x,y}^\xi(i_1,i_2,j_1,j_2) = \ind{\xi((\vec q_{i_1}(x))_{i_2}) =
\xi((\vec q_{j_1}(y))_{j_2})}$. Since there are at most $k$ equality codes in each label, there are
at most $k^2$ equality comparisons. By the union bound, the probability that any of these
comparisons have
\[
\ind{\xi((\vec q_{i_1}(x))_{i_2}) = \xi((\vec q_{j_1}(y))_{j_2})} 
\neq \ind{(\vec q_{i_1}(x))_{i_2} = (\vec q_{j_1}(y))_{j_2}} 
\]
is at most $k^2 \cdot (1/w) = 1/3$, so with probability at least $2/3$ all of the comparisons made
by the decoder have the correct value, so the decoder will be correct. Note that when $(\vec
q_{i_1}(x))_{i_2} = (\vec q_{j_1}(y))_{j_2}$, the random values under $\xi$ will be equal with
certainty. We conclude from this that disjunctive schemes will produce sketches with one-sided
error.
\end{proof}

\begin{rem}
\label{remark:lsh}
Disjunctive labelling schemes with $s=0$ (\ie the $p$ values are empty) can be transformed into
\emph{locality-sensitive hashes (LSH)}~\cite{IM98}. A $(r_1,r_2,\gamma_1,\gamma_2)$-LSH must map any
two points $x,y$ with $\dist(x,y) \leq r_1$ to the same hash value with probability at least
$\gamma_1$, and map any two points $x,y$ with $\dist(x,y) > r_2$ to the same hash value with
probability at most $\gamma_2$, where $r_1 < r_2$ and $\gamma_1 > \gamma_2$. By boosting the success
probability of each \textsc{Equality} check in the disjunction, and then sampling a uniformly random
term from the disjunction, one obtains an LSH with distance parameters that depend on the original
sketch. All of the equality-based sketches presented in this paper, except the first-order sketches,
are of this form.
\end{rem}

\section{Adjacency Sketching}
\label{section:adjacency}

In this section, we prove Theorem~\ref{thm:intro adjacency}, and include the additional equivalent
statement that $\F$ admits a constant-size \emph{disjunctive} adjacency sketch. We think of
disjunctive sketches as the simplest possible use of randomization in a sketch, with the theorem
establishing that the simplest possible sketches are sufficient for monotone classes.

\begin{thm}
\label{thm:adjacency sketching}
Let $\F$ be a monotone class of graphs. Then the following are equivalent:
\begin{enumerate}
\item $\F$ is adjacency sketchable.
\item $\F$ admits a constant-size disjunctive adjacency labelling scheme.
\item $\F$ has bounded arboricity.
\end{enumerate}
\end{thm}
\noindent
A disjunctive labelling scheme for graphs of arboricity at most $k$ can be obtained from the adjacency
labelling scheme of \cite{KNR92}, as in Example~\ref{example:disjunctive}. This leads to a sketch of
size $O(k \log k)$ by Proposition~\ref{prop:eqbased}, which was improved slightly in
\cite{HWZ21}:
\begin{proposition}[\cite{HWZ21}]\label{prop:ask}
Let $\F$ be any class with arboricity at most $k$. Then $\F$ admits a $(0,1,k+1)$-disjunctive
adjacency labelling scheme, and an adjacency sketch of size $O(k)$.
\end{proposition}
\noindent
Therefore, to prove Theorem~\ref{thm:adjacency sketching}, it suffices to prove $(1) \implies (3)$,
which we will prove by contrapositive.
Our proof is inspired by the recent proof of Hambardzumyan, Hatami, and
Hatami \cite{HHH21b}, which refuted a conjecture of \cite{HWZ21} (see Conjecture~\ref{conj:hwz}
below). Our proof also leads to another, more natural counterexample to the conjecture of
\cite{HWZ21}: the class of subgraphs of the hypercube (Corollary~\ref{cor:hypercube
counterexample}).



\noindent
A \emph{spanning subgraph} of a graph $G = (V,E)$ is a subgraph of $G$ with vertex set $V$.  Our
next lemma will give a lower bound on the adjacency sketch size for the class $\G$ of spanning
subgraphs of a graph $G$ of minimum degree $d$. We will actually prove the lower bound for a weaker
type of adjacency sketch, which is only required to be correct on pairs $(x,y)$ that were originally
edges in $G$. This stronger statement is not necessary for the current section, but will be used in
the proof of Theorem~\ref{thm:r2rexpansion}.

For a graph $G = (V,E)$ and the class $\G$ of spanning subgraphs of $G$, and any subgraph $H \in
\G$, we will define the partial function $\mathrm{adj}^E_H : V \times V \to \{0,1,*\}$ as
\[
  \mathrm{adj}^E_H(x,y) = \begin{cases}
    \mathrm{adj}_H(x,y) &\text{ if } (x,y) \in E \\
    * &\text{ otherwise.}
  \end{cases}
\]
In the remainder of this section, we view $\mathrm{adj}^E$ as the
function $(\mathrm{adj}^E_H)_{H\in \G}$ parameterized by $H\in \G$. In
particular, an $\mathrm{adj}^E$-sketch for
$\G$ computes  the partial function $\mathrm{adj}^E_H$ for each $H\in \G$.

We thank an anonymous reviewer for suggesting a proof of the following lemma, which simplified and
improved our original discrepancy argument.

\smallskip

\begin{lemma}\label{lemma:min degree}
Let $G = (V,E)$ be a graph of minimum degree $d$, and let $\G$ be the class of spanning subgraphs of
$G$. Then any $\mathrm{adj}^E$-sketch for $\G$ requires size at least $\Omega(d)$.
\end{lemma}
\begin{proof}
Fix a constant $\delta > 0$ small enough that $\delta\log(e/\delta) \leq 1/2$.  Let $n = |V|$ and
$m=|E|$; we will identify $V$ with $[n]$.  Assume there is an
$\mathrm{adj}^E$-sketch  for $\G$ of size
$s'$; then by standard boosting techniques, there is a sketch of size $s = O(s')$ with error
probability $\delta$ instead of $2/3$. Let $D : \zo^s \times \zo^s \to \zo$ be the decoder for the
sketch.

For every subgraph $H \in \G$, we say a string $\rho \in \zo^{sn}$ is \emph{good for $H$} if the
following holds. Partition $[n] = S_1 \cup \dotsm \cup S_n$ such that $S_i \define \{ (i-1)s+1,
(i-1)s+1, \dotsc, is\}$ is the $i^{th}$ interval of $s$ consecutive indices, and write $\rho_i \in
\zo^s$ for the substring of $\rho$ on indices $S_i$. Then $\rho$ is good for $H$ if
\[
  \left|\left\{ (i,j) : i < j, ij \in E,  D(\rho_i, \rho_j) = H(i,j) \right\}\right|\geq (1-\delta) |E| \,,
\]
where $H(i,j) \in \zo$ is $1$ if and only if $i,j$ are adjacent in $H$. We make two observations:
\begin{enumerate}
\item For every $H \in \G$ there exists $\rho \in \zo^{sn}$ that is good for $H$. Consider the
random string ${\bf \rho}^H \define (\sk(1), \sk(2), \dotsc, \sk(n))$, where $\sk$ is the random
$\mathrm{adj}^E$-sketch of size $s$ for $H$, so that ${\bf \rho}^H \in \zo^{sn}$ with probability 1. Then,
by definition,
\[
  \mathbb{E}\left[ \left|\left\{ (i,j) : i < j, ij\in E,  D({\bf \rho}^H_i, {\bf \rho}^H_j)
      = H(i,j) \right\}\right| \right]
  \geq (1-\delta) |E| \,,
\]
so ${\bf \rho}^H$ is good for $H$ with nonzero probability.  \item Each string $\rho \in \zo^{sn}$
is good for at most $2^{m/2}$ graphs $H \in \G$. This is because, for every string $\rho \in
\zo^{sn}$, we can define $F \in \G$ as the graph obtained by putting each pair $\{i,j\} \in E$, with
$i < j$, adjacent in $F$ if and only if $D(\rho_i, \rho_j) = 1$. If $\rho$ is good for $H$, then
$H(i,j) \neq F(i,j)$ for at most $\delta |E|$ pairs $\{i,j\} \in E$.  Therefore the number of graphs
$H$ such that $\rho$ is good for $H$ is at most
\[
  {m \choose \delta m} \leq
\left(\frac{em}{\delta m}\right)^{\delta m} = 2^{\delta m \log(e/\delta)} \leq 2^{m/2} \,,
\]
by definition of $\delta$.
\end{enumerate}
From these two observations, we have $|\G| \leq 2^{sn} \cdot 2^{m/2}$. Since $|\G| \geq 2^m$, we
conclude that $sn \geq m/2 \geq dn/4$, so $s = \Omega(d)$ and therefore $s' = \Omega(d)$.
\end{proof}

\noindent
We may now complete the proof of Theorem~\ref{thm:adjacency sketching}. We aim to prove $(1)
\implies (3)$, which we will prove by contrapositive: \ie that any class of unbounded arboricity
has non-constant adjacency sketch size. We will prove the following
more precise version, which also completes the proof of Theorem~\ref{thm:intro adjacency}.
\begin{lemma}\label{lem:ask}
  If $\F$ contains a graph of arboricity $d$ and all its subgraphs, then $\F$ requires adjacency sketches of size
  $\Omega(d)$. In particular, if $\F$ is a monotone class of graphs with
  unbounded arboricity, then $\F$ does not admit a
constant-size adjacency sketch.
\end{lemma}
\begin{proof}
It is well-known that the degeneracy of a graph is within factor 2 of
the arboricity, so $\F$ contains a graph $G$ of degeneracy
$\Omega(d)$. By definition, $G$ contains a subgraph $H$ of minimum
degree $\Omega(d)$. Let $\G$ be the class of spanning subgraphs of
$H$. Since $\F$ contains all subgraphs of $G$, it contains all
subgraphs of $H$ and thus $\G \subseteq \F$. Then by Lemma~\ref{lemma:min degree}, any adjacency sketch for $\G$ must
have size $\Omega(d)$, which proves the first part of the result.

If $\F$ is a monotone class of graphs with
  unbounded arboricity, it follows from the paragraph above that for every integer $d$ there is a lower bound of
$\Omega(d)$ on the size of an adjacency sketch for $\F$; it
follows that any adjacency sketch for $\F$ is
of non-constant size.
\end{proof}

Theorem~\ref{thm:intro adjacency} now follows directly from Proposition~\ref{prop:ask} and Lemma~\ref{lem:ask}.

\noindent
As a consequence, we obtain a new counterexample to the following conjecture of \cite{HWZ21}. In
their terminology, a graph class $\F$ is \emph{stable} if there is an absolute constant $k$ such
that any sequence $a_1, b_1, a_2, b_2, \dotsc, a_t, b_t$ of vertices in a graph $G \in \F$, which
satisfies the condition that $a_i, b_j$ are adjacent if and only if $i \leq j$, has length at most
$t \leq k$.

\begin{conj}
\label{conj:hwz}
Let $\F$ be a hereditary graph class which contains at most $2^{O(n \log n)}$
graphs on $n$ vertices, and is stable, Then $\F$ admits a constant-size adjacency sketch.
\end{conj}

\noindent
We remind the reader that the conjecture was already refuted in \cite{HHH21a}, using an interesting
construction of a graph class that was originally used to establish a ``proof barrier'' in
communication complexity \cite{HHH21b}. Our counterexample, the subgraphs of the hypercube, is more
easily defined. The following bound on the number of subgraphs of the hypercube was observed by
Viktor Zamaraev (personal communication). See \cite{HWZ21} for a definition of \emph{stable}.

\begin{corollary}
\label{cor:hypercube counterexample}
Let $\F$ be a class of subgraphs of the hypercube. Then:
\begin{enumerate}
\item $\F$ is stable, and there are at most $2^{O(n \log n)}$ graphs on $n$ vertices in $\F$.
\item $\F$ is not adjacency sketchable.
\end{enumerate}
\end{corollary}
\begin{proof}
Since the $d$-dimensional hypercube of size $N = 2^d$ has minimum degree $d = \log N$, $\F$ has
non-constant adjacency sketch size. To bound the number of $n$-vertex subgraphs of the hypercubes,
we first observe that there are at most $2^{O(n\log n)}$ \emph{induced} subgraphs of the hypercube
on $n$ vertices, which follows from the $O(\log n)$ adjacency labelling scheme for this class
\cite{Har20} (see a simpler exposition at \cite{Har22}).  It is known that any $n$-vertex induced
subgraph of the hypercube has at most $O(n \log n)$ edges \cite{Gra70}, so each induced subgraph
admits at most $2^{O(n\log n)}$ spanning subgraphs. Therefore the number of $n$-vertex subgraphs of
the hypercube is at most $2^{O(n \log n)} \cdot 2^{O(n \log n)} = 2^{O(n \log n)}$.  Any monotone
class of graphs which is not stable contains $K_{t,t}$, for every $t \in \bN$, and therefore
contains the class of all bipartite graphs.  This does not hold for $\F$ (or indeed for any class
containing at most $2^{O(n \log n)}$ $n$-vertex graphs), so $\F$ must be stable.
\end{proof}

\section{Small-Distance Sketching}
\label{section:small-distance}

In this section we prove Theorem~\ref{thm:intro small-distance}. This requires the notion of
\emph{bounded expansion} which we define in Section \ref{section:bounded-expansion} before stating
the formal version of the theorem in Section~\ref{section:small-distance-theorem} and proving it in
the remainder of the section.

\subsection{Bounded expansion}
\label{section:bounded-expansion}

Here we introduce the notion of expansion from sparsity theory, as
discussed in \cite[Chapter 5]{NO12}. We will require some
equivalence results stated in Theorem~\ref{thm:expansion equivalents}.

\begin{definition}[Bounded Expansion]\label{def:boundedexp}
Given a graph $G$ and an integer $r\ge 0$, a \emph{depth-$r$ minor} of $G$ is a graph obtained by
contracting pairwise disjoint  connected subgraphs of radius at most $r$ in a subgraph of $G$.  For
any function $f$, we say that a class of graphs $\mathcal{G}$ has \emph{expansion} at most $f$ if
any depth-$r$ minor of a graph of $\mathcal{G}$ has average degree at
most $f(r)$ (see~\cite[Section 5.5]{NO12}
for more details on this notion).  We say that a class $\mathcal{G}$ has \emph{bounded expansion} if
there is a function $f$ such that $\mathcal{G}$  has expansion at most $f$.
\end{definition}

Note that, for example, every proper minor-closed family has constant
expansion. We now introduce two other equivalent ways to define bounded
expansion:   via generalized coloring numbers, and via bounded depth
topological minors (both we be useful for our purposes). 

\begin{definition}[Weak $r$-coloring number]
\label{def:weak reachability}
Given a total order $(V,<)$ on the vertex set $V$ of a graph $G$ and
an integer $r\ge 0$, we
say that a vertex  $v\in V$ is \emph{weakly $r$-reachable} from a vertex
$u\in V$ if there is a path of length at most $r$ connecting $v$ to
$u$ in $G$, and such that for any vertex $w$ on the path, $v\le w$ (in
words, $v$ is the smallest vertex on the path with respect to
$(V,<)$). For a graph $G$ and an integer $r\ge 0$, the \emph{weak
  $r$-coloring number}
$\wcol_r(G)$ is the smallest integer $k$ for which the vertex set of $G$ has a total order $(V,<)$ such
that for any vertex $u\in V$, at most $k$ vertices are weakly $r$-reachable
from $u$ with respect to $(V,<)$. For a graph class $\mathcal{F}$, we
write $\wcol_r(\mathcal{F})$ for the supremum of $\wcol_r(G)$, for $G \in \mathcal{F}$.
\end{definition}

\begin{definition}[$(k,\ell)$-Subdivisions]
For a graph $G$ and two integers $0\le k\le \ell$, a $(k,
\ell)$-subdivision of $G$ is any graph obtained from $G$ by
subdividing each edge of $G$ at least $k$ times and at most $\ell$
times (\ie we replace each edge of $G$ by a path with at least $k$
and an most $\ell$ internal vertices). A $(k,k)$-subdivision is also
called a \emph{$k$-subdivision} for simplicity;
\end{definition}

\begin{definition}[Depth-$r$ Topological Minor]
We say that $H$ is a \emph{depth-$r$ topological minor} of a graph $G$
if $G$ contains a $(0, 2r)$-subdivision of $H$ as a subgraph. In other
words, $G$ contains a subgraph $H'$ obtained from $H$ by subdividing
each edge at most $2r$ times.
\end{definition}
In the proof below it will be convenient to use the following equivalent definitions of bounded
expansion. 

\begin{thm}[\cite{NO12}]\label{thm:expansion equivalents}\label{thm:no}\label{thm:be}
For a class $\mathcal{F}$ of graphs, the following are equivalent:
\begin{enumerate}
\item $\mathcal F$ has bounded expansion.
\item \label{thm:ee wcol}
There is a function $f : \mathbb N \to \mathbb N$ such that for any $r \in \mathbb N$,
$\wcol_r(\mathcal F) \leq f(r)$.
\item \label{thm:ee top}  There is a function $f : \mathbb N \to \mathbb N$ such that for any $r \in
\mathbb N$ and any $G \in \mathcal F$, any depth-$r$ topological minor of $G$ has average degree at
most $f(r)$.
\end{enumerate}
\end{thm}
\noindent
The equivalence between 1.\ and 3.\ follows from Proposition
5.5 in~\cite{NO12}; while the equivalence between 1.\ and 2.\ follows
from Lemma 7.11 and Theorem 7.11 in~\cite{NO12}.

\noindent
We will also require the following fact about the expansion of monotone classes, which is a simple
consequence of Theorem~\ref{thm:no} (see for instance \cite{NO15}) combined with a result of K\"uhn
and Osthus~\cite{KO04}.

\begin{corollary}\label{cor:no}
Let $\mathcal{F}$ be a monotone class of unbounded expansion. Then
there is a constant $r\ge 0$, so that for any $d\ge 0$, $\mathcal{F}$
contains an $r$-subdivision of a bipartite graph of minimum degree at
least $d$ and girth at least 6.
\end{corollary}

\begin{proof}
Since $\mathcal{F}$ has unbounded expansion, it follows from
Theorem~\ref{thm:no} that there exists an $r\ge 0$,
such that depth-$r$ topological minors of graphs in $\mathcal{F}$ have unbounded
average degree. Since each edge in a depth-$r$ topological minor is
subdivided at most $2r$ times, if $\mathcal{F}$ contains a depth-$r$
topological minor $H$ of average degree at least $d$, $\mathcal{F}$
also contains an $r'$-subdivision of a subgraph $H'$ of $H$, for some
$r'\le 2r$, such that $H'$ has average degree at least
$\tfrac{d}{2r+1}$ (recall that $\mathcal{F}$ is  monotone). It follows
that there exists an integer $r''\le 2r$ such that for infinitely many
$d$, $\mathcal{F}$
contains an $r''$-subdivision of a graph of average degree at least
$d$. It was proved by K\"uhn and Osthus~\cite{KO04} that any graph of
sufficiently  large average degree
contains a bipartite subgraph of large minimum
degree and girth at least 6. As $\mathcal{F}$ is  monotone, the desired result follows.
\end{proof}

We should remark that the weaker version of Corollary \ref{cor:no} where the
girth at least 6 is replaced by girth at least 4 is much simpler and
does not require the result of K\"uhn and Osthus~\cite{KO04}: it
suffices to use the simple result that any graph $G$ of large average
degree contains a bipartite graph of large average degree as a
subgraph (consider for instance a random bipartition of $G$).

\subsection{Statement of Theorem \ref{thm:intro small-distance}}
\label{section:small-distance-theorem}

As in Theorem~\ref{thm:adjacency sketching} from the previous section, we refine the theorem by
showing that the sketches are in fact disjunctive.

\begin{thm}
\label{thm:small-distance sketching}
Let $\F$ be a monotone class of graphs. Then the following are equivalent:
\begin{enumerate}
\item $\F$ is small-distance sketchable.
\item For some function $f : \bN \to \bN$ and every $r \in \bN$, $\F$ admits a disjunctive
small-distance labelling scheme of size $f(r)$.
\item $\F$ is first-order sketchable. 
\item $\F$ has bounded expansion.
\end{enumerate}
\end{thm}

\noindent
It holds by definition that $(3) \implies (1)$ and $(2) \implies (1)$, even without the assumption
of monotonicity. We will prove $(4) \implies (3)$ and $(4) \implies (2)$ using different methods.
We prove $(4) \implies (3)$ (again without the assumption of monotonicity) in
Section~\ref{section:first-order sketching} using the structural result of \cite{GKN+20}. This proof
does not give explicit bounds on the sketch size. $(4) \implies (2)$ is proved in
Section~\ref{section:small-distance upper bounds} and gives explicit upper bounds on the sketch
size.  The final piece of the theorem, $(1) \implies (4)$, is proved in
Section~\ref{section:small-distance lower bounds}. 

In Section~\ref{section:large r}, we present small-distance sketches for graphs embeddable on a
given surface, with sketch size that has a smaller dependence on $r$ than in
Section~\ref{section:small-distance upper bounds}, but some dependence on the graph size $n$. This
is desirable when $r$ is large (as a function of $n$).

\subsection{Bounded Expansion Implies FO Labelling Schemes}
\label{section:first-order sketching}

To prove that any class of bounded expansion is first-order sketchable, we use the result of
\cite{GKN+20} that shows how to decompose any class of (structurally) bounded expansion into a
number of graphs of bounded shrubdepth. We will require an adjacency
sketch for classes of bounded shrubdepth,
given below.

\subsubsection{Adjacency Sketching for Bounded Shrubdepth}

We must first define shrubdepth.
A {\em{connection model}} for a graph $G$ is a rooted tree $T$ whose nodes are colored with some
number $k$ of colors, such that:
\begin{itemize}
\item the vertices of $G$ are  the leaves of $T$; and
\item for two vertices $u,v\in V(G)$, whether $u$ and $v$ are
adjacent in $G$ depends only on the colors of $u$ and $v$  in~$T$, and the color of the lowest
common ancestor of $u$ and $v$ in $T$.
\end{itemize}
To avoid ambiguity, we say $G$ has \emph{vertices} while $T$ has \emph{nodes}.
Note that we can assume without loss of generality that all leaves are at the same distance from the
root in $T$.  A class $\mathcal{G}$ has \emph{bounded shrubdepth} if there are some $d,k\in
\mathbb{N}$ such that every $G\in\mathcal{G}$ has a connection model of depth $d$ with colors in
$[k]$ (we recall that the \emph{depth} of a rooted tree $T$ is the
maximum number of edges on a root-to-leaf path in $T$). The reader is
referred to \cite{GHNOMR12,GHNOM19} for more details on shrubdepth, and a number
algorithmic applications in model checking.

\begin{lemma}\label{adj:shrubdepth}
\label{lemma:shrubdepth}
Any class $\mathcal G$ of bounded shrubdepth admits a constant-size equality-based adjacency
labelling scheme.
\end{lemma}
\begin{proof}
Let $d,k$ be such that any graph $G\in \mathcal{G}$ has a connection model $T_G$ of depth $d$ using
color set $[k]$. We denote by $\varphi_G:[k]^3\to \{0,1\}$ the function such that if $u$ has color
$a$, $v$ has color $b$, and the lowest common ancestor of $u$ and $v$ has color $c$ in $T_G$, then
$u$ and $v$ are adjacent in $G$ if and only if $\varphi_G(a,b,c)=1$. For every node $u$ of $T_G$,
write $\chi(u)$ for the color of $u$ in the connection model.

We now construct our equality-based labels for $G$. For any vertex $x$, let $t_0(x), t_1(x), \dotsc,
t_d(x)$ be the leaf-to-root path for $x$, where $t_0(x)=x$ and $t_d(x)$ is the root of $T_G$. Then the
label for $x$ is the sequence $(\varphi_G \mid -), (\chi(t_0(x)) \mid t_0(x)), \dotsc, (\chi(t_d(x))
\mid t_d(x))$.

On inputs
\begin{align*}
(\varphi_G \mid -), (\chi(t_0(x)) \mid t_0(x)), \dotsc, (\chi(t_d(x)) \mid t_d(x)) \,,\\
(\varphi_G \mid -), (\chi(t_0(y)) \mid t_0(y)), \dotsc, (\chi(t_d(y)) \mid t_d(y)) \,,
\end{align*}
the decoder operates as follows. It finds the smallest $i \in [d]$ such that $\ind{t_i(x)=t_i(y)}$
and outputs $\varphi_G(\chi(t_1(x)), \chi(t_1(y)), \chi(t_i(x)))$.

The correctness of this labelling scheme follows from the fact that we will have $t_i(x) = t_i(y)$
if and only if the node $t_i(x) = t_i(y)$ is an ancestor of both $x$ and $y$ in $T_G$, so the smallest
$i \in [d]$ such that $t_i(x)=t_i(y)$ identifies the lowest common
ancestor of $x$ and $y$  in $T_G$.
\end{proof}

\subsubsection{Structurally Bounded Expansion Implies First-Order Sketching}

Following \cite{GKN+20}, we say that a class of graphs has \emph{structurally bounded expansion} if
it can be obtained from a class of bounded expansion by first-order
(FO) transductions. We omit the precise definition of FO transductions
in this paper, as they are not necessary to our discussion, and instead
refer the reader to \cite{GKN+20}. We just note that a particular case of
FO transduction is the notion of \emph{FO interpretation}, which is of
specific interest to us. Consider an FO formula $\phi(x,y)$ with two free variables and relational vocabulary $\Sigma = \{ F,
R_1, \dotsc, R_k \}$ where $F$ is symmetric of arity 2. We will say
that a graph class $\F'$ is an FO interpretation of a graph class $\F$
with respect to $\phi$
if for any graph $G'=(V,E')\in \F'$ there is a graph $G=(V,E)\in \F$ and a $\Sigma$-structure with domain $V$ where $E$ is the interpretation of the
symbol $F$, such that for any pair $u,v\in V$, $uv\in E'$ if and only
if $\phi(u / x, v / y)$ evaluates to true. For instance, if
$\phi(u/x,v/y)$ encodes the property  $\dist_G(u,v)\le r$ for some fixed
integer $r\ge 1$ (which can
be written as an FO formula), then the corresponding FO interpretation of the class
$\F$ is the class of all graph powers $\{G^r\,|\, G\in F\}$. FO
transductions are slightly more involved, as it is allowed to consider
a bounded number of copies of a graph before applying the formula, and
then it is possible to delete vertices. We will use
the following structural result for classes of structurally bounded expansion, proved in
\cite{GKN+20}.

\begin{thm}[\cite{GKN+20}]\label{thm:gkn+20}
 A class $\mathcal{G}$ of graphs  has structurally bounded expansion if and only if the following condition holds.
  For every $p\in \mathbb{N}$, there is a constant $m=m(p)$ such that for every graph $G\in \mathcal{G}$, one can find a family $\mathcal{F}(G)$ of vertex subsets of $G$ with $|\mathcal{F}(G)|\leq m$ and the following properties:
 \begin{itemize}
  \item for every $X\subseteq V(G)$ with $|X|\leq p$, there is $A\in \mathcal{F}(G)$ such that $X\subseteq A$; and
  \item the class $\{G[A]\,|\, G\in \mathcal{G}, A\in \mathcal{F}(G)\}$ of induced subgraphs has bounded shrubdepth.
 \end{itemize}
\end{thm}

We directly deduce the following result.

\begin{lemma}\label{adj:sbe}
\label{lemma:sbe}
Any class $\mathcal{G}$ of structurally bounded expansion admits a constant-size equality-based
adjacency labelling scheme.
\end{lemma}

\begin{proof}
Let $m$ and $\mathcal{F}$ be given by applying Theorem \ref{thm:gkn+20} to $\mathcal{G}$ with $p=2$.
By definition, for every graph $G \in \mathcal G$ and every pair of vertices $u,v \in V(G)$, there
is a set $A \in \mathcal{F}(G)$ containing $u$ and $v$. Moreover, $\mathcal{F}(G)$ contains at most
$m$ sets and the family $\mathcal{C}$ of all graphs $G[A]$, for $G\in \mathcal{G}$, and $A\in
\mathcal{F}(G)$, has bounded shrubdepth. It follows from Lemma~\ref{lemma:shrubdepth} that there is
a constant-size equality-based adjacency labelling scheme for $\mathcal C$. We denote the decoder of
this scheme by $D$, and the corresponding labels as $\ell_{G[A]}$.

Consider some graph $G\in \mathcal{G}$, and let
$\mathcal{F}(G)=\{A_1,\ldots,A_m\}$ (for convenience we consider
$\mathcal{F}(G)$ as a multiset of size exactly $m$, that is, some sets $A_i$
might be repeated). For each vertex
$x$ of $G$ and $i \in [m]$, we write $a(x) = (a_1(x), \dotsc, a_m(x))$ where $a_i(x) = \ind{x \in
A_i}$. Then we define the label for $x$ by taking the prefix $a(x)$ and appending the labels
$\ell_{G[A_i]}(x)$ for each induced subgraph $G[A_i] \in \mathcal C$ to which $x$ belongs. Given
the labels for vertices $x$ and $y$, the decoder finds any $i \in [m]$ such that $a_i(x)=a_i(y)=1$;
and outputs $D(\ell_{G[A_i]}(x), \ell_{G[A_i]}(y))$. Such a number $i \in [m]$ always exists due
to Theorem~\ref{thm:gkn+20}.  The correctness of this labelling scheme follows from
Theorem~\ref{thm:gkn+20} and Lemma~\ref{lemma:shrubdepth}.
\end{proof}

Since FO-transductions compose (see e.g.~\cite{NOS22}), sketching FO formulas in a class of
structurally bounded expansion is equivalent to sketching adjacency in another class of structurally
bounded expansion. We obtain the following direct corollary of Theorem~\ref{adj:sbe}.

  \begin{corollary}\label{cor:fosketch}
Any class $\mathcal{G}$ of structurally bounded expansion is
first-order sketchable.
  \end{corollary}
  
As the property $\dist_G(x,y)\le r$ can be written as an FO formula, this
directly implies that classes of bounded expansion are small-distance
sketchable.  However, this does not tell anything on the size of the
sketches as a function of $r$, unlike the approach using weak coloring
numbers described in the next section.


\newcommand{\sign}{\mathrm{sign}}
\newcommand{\rank}{\mathrm{rank}}
\newcommand{\inn}[1]{\left\langle #1 \right\rangle}
We also observe another interesting corollary, that graph classes of structurally bounded expansion
can be represented as constant-dimensional geometric intersection graphs. This follows from our
lemma in conjunction with concurrent work of \cite{HHP+22}. We require the notion of
\emph{sign-rank}. The \emph{sign-rank} $\mathrm{rank}_\pm(M)$ of a sign matrix $M \in \{\pm 1\}^{n
\times n}$ is the minimum rank $r$ of a matrix $R \in \R^{n \times n}$ such that $M(i,j) =
\mathrm{sign}(R_{i,j})$ for all $i,j \in [n]$.  Equivalently, $\mathrm{rank}_\pm(M)$ is the minimum
dimension $d$ such that $M$ can be represented as a $d$-dimensional point-halfspace arrangement with
halfspaces through the origin; \ie $M_{i,j} = \sign(\inn{u_i, v_j})$ where $u_i, v_j$ are unit
vectors assigned to each row $i$ and column $j$.

For any graph $G$ on $n$ vertices, we can define its \emph{adjacency sign-matrix} as the matrix $M
\in \{\pm 1\}^{n \times n}$ with $M_{i,j} = 1$ if and only if $i,j$ are adjacent in $G$. We say that
a class $\G$ of graphs has \emph{bounded sign-rank} if there exists a constant $r$ such that for
every $G \in \G$, the sign-rank of the adjacency sign-matrix of $G$ is at most $r$. A graph class
$\G$ with bounded sign-rank is therefore a type of constant-dimensional geometric intersection
graph; if $\G$ has bounded sign-rank, then $\G$ is \emph{semi-algebraic}, (and it is not known
whether having bounded sign-rank is \emph{equivalent} to being semi-algebraic \cite{HHP+22}).

\begin{corollary}
\label{cor:sign-rank}
Any class $\mathcal{G}$ of structurally bounded expansion has bounded sign-rank.
\end{corollary}
\begin{proof}
A result of \cite{HHP+22} is that for every $M \in \{\pm 1\}^{n \times n}$, $\rank_\pm(M) \leq
4^{\mathsf{D}^{\textsc{Eq}}(M)}$, where $\mathsf{D}^{\textsc{Eq}(M)}$ is the deterministic
communication cost of $M$ with access to an \textsc{Equality} oracle. $\G$ having a
constant-size equality-based adjacency labelling scheme is equivalent to the existence of a constant $c$
such that $\mathsf{D}^{\textsc{Eq}}(M) \leq c$ for every adjacency sign-matrix $M$ of a graph in
$\G$. Therefore every adjacency sign-matrix $M$ for graphs $G \in \G$ has $\rank_\pm(M) \leq 4^c$.
\end{proof}

\medskip

\begin{rem}
 A very recent result of Dvo\v r\'ak~\cite[Corollary
9]{Dvo22}, which appeared after our paper was made public, can also be
used to give an alternative proof of the result that any class of structurally
bounded expansion is small-distance sketchable. The drawback is that
his result applies to small-distance sketching (rather than
first-order sketching), but the benefit is that the small-distance
sketch for a vertex $v$ in a given graph $G$ can be
obtained by taking a constant number of adjacency sketches of the
vertices in a small ball around $v$ in $G$. This approach is more
general than the approach we take in the next section, which considers
classes of bounded expansion instead of classes of structurally
bounded expansion.
\end{rem}

\subsection{Bounded Expansion Implies Small-Distance Sketching}
\label{section:small-distance upper bounds}

Recall the definition of weak $r$-coloring number from Definition~\ref{def:weak reachability}. We give a
quantitative bound on the small-distance sketch of any graph class $\mathcal F$ in terms of
$\wcol_r(\mathcal F)$. Recall from Theorem~\ref{thm:expansion equivalents} that any class with
bounded expansion has $\wcol_r(\mathcal F) \leq f(r)$ for some function $f(r)$; therefore we obtain
the existence of small-distance sketches for any class of bounded expansion.

\begin{thm}\label{thm:distr1}
For any $r \in \bN$, any class $\F$ has an $(0, r, \wcol_r(\F))$-disjunctive
distance-$(r,r)$ labelling scheme.
\end{thm}
\begin{proof}
Let $G \in \F$, and consider a total order $(V, \prec)$ such that for any vertex $x \in V$, at most
$\wcol_r(\F)$ vertices are weakly $r$-reachable from $v$ in $G$ with respect to $(V,\prec)$. We say
that vertex $y \in V$ has \emph{$x$-rank} $k$ if $y$ is weakly $k$-reachable from $x$ but not weakly
$(k-1)$-reachable from $x$. For each vertex $x$ and $k \in [r]$, write $S_k(x)$ for the set of
vertices $y$ with $x$-rank $k$.

We construct a disjunctive labelling scheme as follows. Each vertex $x$ is assigned the label
\[
  (- \mid \vec q_1(x)), (- \mid \vec q_2(x)), \dotsc, (- \mid \vec q_{r'}(x))
\]
where $r' \leq r$ is the maximum number such that $S_{r'}(x) \neq \emptyset$, and the equality codes
$\vec q_i(x)$ are names of vertices in the set $S_i(x)$. Each label contains at most $\wcol_r(G)$
equality codes. Given labels for $x$ and $y$, the decoder outputs 1 if and only if there exist $0
\leq i,j \leq r$ such that $i+j \leq r$ and $S_i(x) \cap S_j(y) \neq \emptyset$, which can be
checked using the equality codes in $\vec q_i(x)$ and $\vec q_j(y)$.

Suppose that $\dist_G(x,y) \leq r$ and let $P \subseteq V(G)$ be a path of length $\dist_G(x,y)$.
Let $z \in P$ be the minimal element of $P$ with respect to $\prec$. Then $z$ is weakly
$i$-reachable from $x$ and weakly $j$-reachable from $y$, for some values $i,j$ such that $i+j \leq
r$. Then $z \in S_i(x) \cap S_j(y)$, so the decoder will output 1 given the labels for $x$ and $y$.
On the other hand, if the decoder outputs 1, then there are values $i,j$ such that $i + j \leq r$
and $S_i(x) \cap S_j(y) \neq \emptyset$. Let $z \in S_i(x) \cap S_j(y)$, so that $z$ is weakly
$i$-reachable from $x$ and weakly $j$-reachable from $y$. Then $\dist_G(x,y) \leq \dist_G(x,z) +
\dist_G(z,y) \leq i + j \leq r$.
\end{proof}

We noticed after proving this result that a similar idea was used in \cite[Lemma
6.10]{GKS17} to obtain sparse neighborhood covers in nowhere-dense classes.

We will need the following quantitative results for planar graphs
and graphs avoiding some specific minor, due to~\cite{HOQRS17}.

\begin{thm}[\cite{HOQRS17}]\label{thm:beplanar}
For any planar graph $G$, and any integer $r\ge 0$, $\wcol_r(G)\le (2r+1){r+2\choose 2}=O(r^3)$.
\end{thm}

\begin{thm}[\cite{HOQRS17}]\label{thm:beminor}
For any integer $t\ge 3$, any graph $G$ with no $K_t$-minor, and any
integer $r\ge 0$, $\wcol_r(G)\le {r+t-2\choose t-2}(t-3)(2r+2)=O(r^{t-1})$.
\end{thm}
\noindent
In the proof of Theorem~\ref{thm:distr1}, the equality codes are just the names of
vertices; so we can use $\lceil \log n \rceil$ bits to encode each of the $\wcol_r(\F)$ equality
codes to obtain an adjacency label. Then, combined with Proposition~\ref{prop:eqbased}, we obtain
the following corollary:
\begin{corollary}\label{cor:distr1}
If a class $\mathcal{F}$ has bounded expansion, then $\F$ has a small-distance sketch of size at
most $O(r + \wcol_r(\F)\log(\wcol_r(\F)))$. If $\F$ is the class of planar graphs, then the sketch
has size $O(r^3 \log r)$ and if $\F$ is the class of $K_t$-minor free graphs for some fixed integer
$t\ge 3$, then the sketch has size $O(r^{t-1}\log r)$. Furthermore, $\F$ admits a distance-$(r,r)$
labelling scheme of size $O(r + \wcol_r(\F)\log n)$; if $\F$ is the class of planar graphs, then the
scheme has size $O(r^3\log n)$ and if $\F$ is the class of $K_t$-minor free graphs, then the scheme
has  size $O(r^{t-1}\log n)$.
\end{corollary}

\begin{rem}
The fact that the sketch size is independent of the number of vertices in Corollary~\ref{cor:distr1}
implies that the scheme actually works for infinite graphs. It was proved in~\cite{HMSTW21} that for
infinite graphs $G$, $\wcol_r(G)$ is the supremum of $\wcol_r(H)$ for all finite subgraphs $H$ of
$G$ (this was actually proved explicitly for the strong coloring numbers instead of the weak
coloring numbers, but the proof is the same). This shows that Theorems~\ref{thm:beplanar} and
\ref{thm:beminor}, and thus Corollary~\ref{cor:distr1}, also hold for infinite graphs.
\end{rem}

\noindent
In Section~\ref{section:large r} below, we will see how to improve the
dependence in $r$ in the result above, at the cost of an increased dependence in $n$.

\subsection{Small-Distance Sketching Implies Bounded Expansion}
\label{section:small-distance lower bounds}

To complete the proof of Theorem~\ref{thm:small-distance sketching}, we must show that any monotone
class of graphs that is small-distance sketchable has bounded expansion, which we do by
contrapositive. In fact, we will prove a stronger statement: even having a weaker
$(r,5r-1)$-distance sketch of size $f(r)$ implies bounded expansion.

\begin{thm}\label{thm:r2rexpansion}
Let $\mathcal{F}$ be a monotone class of graphs and assume that there is a function $f$ such that
for any $r\ge 1$, $\mathcal{F}$ has a $(r,5r-1)$-distance sketch of size $f(r)$.  Then $\mathcal{F}$
has bounded expansion.
\end{thm}

\begin{proof}
Assume for the sake of contradiction that $\mathcal{F}$ has unbounded expansion.  By Corollary
\ref{cor:no}, there is a constant $k$ such that for every $d\ge 0$, $\mathcal{F}$ contains a
$k$-subdivision of some bipartite graph $G=(V,E)$ of minimum degree at least $d$ and girth at least
6. Let $\mathcal{G}$ be the class consisting of the graph $G$, together with all its spanning
subgraphs. By monotonicity, $\mathcal{F}$ contains $k$-subdivisions of all the graphs of
$\mathcal{G}$.

Recall the definition of the partial function $\mathrm{adj}^E$ parameterized by graphs $H \in \G$,
from the discussion preceding Lemma~\ref{lemma:min degree}. We will show that the $(k+1,
5(k+1)-1)$-distance sketch of size $f(k+1)$ for $\F$ can be used to obtain an $\mathrm{adj}^E$-sketch
for $\G$, which must have size $\Omega(d)$ due to Lemma~\ref{lemma:min degree}. This is a
contradiction since we must have $f(k) = \Omega(d)$ for arbitrarily large $d$, whereas $f(k+1)$
is a constant independent of $d$.

Let $H$ be any spanning subgraph of $G$ and let $H^{(k)}$ denote the $k$-subdivision of $H$.
Consider two vertices $u,v \in V(H) \subseteq V(G)$ that are adjacent in $G$.  Observe that
$\dist_{H^{(k)}}(u,v)=(k+1) \dist_H(u,v)$, and thus if $u,v$ are adjacent in $H$ then
$\dist_{H^{(k)}}(u,v)\le k+1$. Assume now that $u,v$ are non-adjacent in $H$. Since $u,v$ are
adjacent in $G$, $G$ has girth at least 6, and $H$ is a spanning subgraph of $G$, it follows that in
this case $\dist_H(u,v)\ge 5$, and thus $\dist_{H^{(k)}}(u,v)\ge 5(k+1)$. Therefore, by using the
same decoder as the $(k+1,5(k+1)-1)$-distance sketch for $\F$, and using the random sketch $\sk$
defined for $G$, we obtain an $\mathrm{adj}^E$-sketch for $H$. This gives an
$\mathrm{adj}^E$-sketch for $\G$ of size $f(k+1)$.
\end{proof}

In our proof of Theorem~\ref{thm:r2rexpansion} we have used
Corollary~\ref{cor:no}, which is based on the result of \cite{KO04}, stating that every graph of
large minimum degree contains a bipartite subgraph of girth at least 6
and large minimum degree. As explained after the proof of
Corollary~\ref{cor:no}, the version of this statement with girth 4
instead of girth 6 is much easier to prove, and using  it gives a version of 
Theorem~\ref{thm:r2rexpansion} with a $(r,3r-1)$-distance sketch
(instead of a $(r,5r-1)$-distance sketch) that does not use the result
of \cite{KO04}.

The following statement was
conjectured by Thomassen~\cite{Tho83}.

\begin{conj}[\cite{Tho83}]\label{conj:thom}
For every integer $k$, every graph of
sufficiently large minimum degree contains a bipartite subgraph of girth at least $k$
and large minimum degree.
\end{conj}

If Conjecture~\ref{conj:thom} is true, it readily follows from our proof that the constant 5 in
Theorem~\ref{thm:r2rexpansion} can be replaced by an arbitrarily large constant. Compare this with
Theorem~\ref{thm:weakwood} in the next section, where we prove this result for randomized labelling
schemes whose label size is constant (independent or $r$).

Theorem~\ref{thm:r2rexpansion} bears some similarities with the following result of Ne\v set\v ril
and Ossona de Mendez~\cite{NO15}. An \emph{$r$-neighborhood cover} of a graph $G$ is a cover of its
vertex sets by connected subgraphs $S_1,S_2,\ldots,S_k$ of radius at most $2r$, such that  any ball
of radius $r$ in $G$ is contained in some set $S_i$. The \emph{multiplicity} of the cover is the
maximum, over all vertices $v$ of $G$, of the number of sets $S_i$ containing $v$. It was proved in
\cite{GKS17} that for every class of bounded expansion there is a function $f$ such that the class
has $r$-neighborhood covers of multiplicity $f(r)$, for any $r$ (our proof of Theorem
\ref{thm:distr1} uses similar arguments, as explained earlier). Ne\v set\v ril and Ossona de
Mendez~\cite{NO15} then proved that the converse holds for monotone classes. Note that the existence
of $r$-neighborhood covers of multiplicity $f(r)$ can be used to design randomized
distance-$(r,4r)$ sketches of size at most $O(f(r))$, but the converse does not
seem to be possible (\ie it does not seem possible to obtain $r$-neighborhood covers from
distance-$(r,4r)$ sketches). 

\subsection{Improved Sketches for Large $r$}
\label{section:large r}

We now explain how to improve the dependence in $r$ (at the cost of an increased dependence in $n$)
in Corollary~\ref{cor:distr1}, in the case of planar graphs, or more generally graphs embedded on a
fixed surface. This result is due to Gwena\"el Joret, who we thank for allowing us to include it
here.
We will require the notion of \emph{layering}:

\begin{definition}[Layering]
A layering in a graph $G$ is a partition $V_1, V_2,\ldots$ of its vertex set (where each set $V_i$
is called a \emph{layer}) such that each edge is either inside some layer $V_i$, or between
consecutive layers, say $V_i$ and $V_{i+1}$ (for some $i$).
\end{definition}
\noindent
We will use the following result due to Robertson and Seymour~\cite{RS84} for planar graphs, and
later extended by Eppstein~\cite{Epp00} to general surfaces. Their
results are stated for  balls of bounded
radius but they easily imply the version
below by considering a Breadth-First Search layering. See also Theorem 12 in \cite{DMW17} for a result that
directly implies Theorem \ref{thm:lay}.

\begin{thm}[\cite{Epp00,RS84}]\label{thm:lay}
For any fixed surface $\Sigma$, there is a constant $c$ such that any
graph embeddable on $\Sigma$ has a layering such that for any integer
$k$, the subgraph induced by any $k$ consecutive layers has treewidth
at most $ck$.
\end{thm}

\noindent
We will also need the following simple result~\cite[Corollary
6.1]{NO12}.

\begin{lemma}[\cite{NO12}]\label{lem:tw}
There is a constant $c>0$ such that if an $n$-vertex graph $G$ has treewidth at most  $k$, then  for
any integer $r\ge 0$, $\wcol_r(G)\le c k\log n$.
\end{lemma}

\begin{thm}\label{thm:distr2}
For any surface $\Sigma$, and integers $r\ge 0$ and $n\ge 1$, the
class of $n$-vertex graphs embeddable on $\Sigma$ has a distance-$(r,r)$
labelling scheme with labels of at most  $O(r \log^2
n)$ bits. 
\end{thm}

\begin{proof}
Consider an $n$-vertex graph $G$ embeddable on $\Sigma$ and an integer
$r\ge0$. By Theorem~\ref{thm:lay}, there is a layering  $V_1, V_2,
\ldots$ of $G$ such that for any integer $k$, any $k$ consecutive
layers induce a graph of treewidth  $O(k)$. For any integer $i\ge 0$, let
$L_i$ be the union of layers $V_{ir+1}, V_{ir+2}, \ldots,
V_{(i+2)r}$. Note that each set $L_i$ induces a graph of
treewidth $O(r)$ and each vertex lies in at most 2 sets $L_i$
(namely $L_i$ and $L_{i+1}$, for some $i$). Moreover, any path  of
length at most $r$ in $G$, is contained in some set $L_i$.

Assign each vertex of $G$ a distinct identifier from $[n]$.
By Lemma~\ref{lem:tw}, each set $L_i$ has a total order such that for
any $v\in L_i$,  at
most $O(r \log n)$ vertices are weakly $r$-reachable from $v$ in
$G[L_i]$. 
Each vertex $v$ records the indices of the (at most 2) sets $L_i$ it
lies in, together with the sets of (identifiers of) the vertices that are weakly
$r$-reachable from $v$ in these (at most 2) sets $L_i$ (as in the
proof of Theorem~\ref{thm:distr1}, we store each
of these vertices together with its $v$-rank with respect to the
ordering of $L_i$, so in total each vertex has a label of at most
$2\log n+O(r \log n)\cdot O(\log n)=O(r\log^2 n)$ bits).

As each path of length at most $r$ is contained in some $L_i$, the
property $d(u,v)\le r$ can be checked as in the proofs of
Theorem~\ref{thm:distr1}, by testing whether in
some $L_i$ in which both $u$ and $v$ lie, some vertex is weakly
$s$-reachable from $u$ and weakly $t$-reachable from $v$ (with respect
to the total ordering of $L_i$), such that $s+t\le r$.
\end{proof}

\noindent
Note that one could replace the identifiers of the weakly $r$-reachable vertices with equality
codes, to obtain a sketch; however, this would save a log factor at
best.

\noindent
As observed by a reviewer, Theorem \ref{thm:distr2} can also be
deduced from Corollary \ref{cor:distr1}, combined with the result that
$n$-vertex graphs of bounded genus have weak $r$-coloring number $O(r
\log n)$. More precisely, in Corollary \ref{cor:distr1} we proved that
any class $\mathcal{F}$ has a distance-$(r,r)$ labelling scheme of
size $O(r+\wcol_r(\mathcal{F})\log n)$ but the proof shows the
stronger result that in such a scheme, every $n$-vertex graph $G$ is assigned
labels of size at most $O(r+\wcol_r(G)\log n)$. Now, it remains to prove
that $n$-vertex graphs of bounded genus have weak $r$-coloring number $O(r
\log n)$. As observed by Gwena\"el Joret (personal communication),
this can be quickly deduced from Lemma \ref{lem:tw} and
\cite[Theorem 20]{DJMMUW20}, in a way that is very similar to the
proof of Theorem \ref{thm:distr2}.

\section{Approximate Distance Threshold Sketching}
\label{section:adt}

In this section, we prove the results listed in Section~\ref{section:intro adt}.  Recall that a
class $\F$ admits an $\alpha$-ADT sketch of size $s(n)$ if for every $r \in \bN$ there is a function
$D_r : \zo^* \times \zo^* \to \zo$ such that every graph $G \in \F$ on $n$ vertices admits a probability
distribution over functions $\sk : V(G) \to \zo^{s(n)}$ such that, for all $x,y \in V(G)$,
\begin{align*}
  \dist_G(u,v) \leq r &\implies \Pr[ D_r(\sk(u), \sk(v)) = 1 ] \geq 2/3 \,,\\
  \dist_G(u,v) > \alpha r &\implies \Pr[ D_r(\sk(u), \sk(v)) = 0 ] \geq 2/3 \,.
\end{align*}
\noindent
We emphasize that the size of the sketch should not depend on $r$, especially when $s(n)$ is constant (unlike in
the notion of small-distance sketches studied in previous sections).  A desirable property of these
sketches is that the decoder $D_r$ does not depend on $r$ either, so that $D_r = D_{r'} = D$ for
every $r,r'$.  We will call such sketches \emph{distance-invariant}. We remark that any $\alpha$-ADT
sketch can be made distance-invariant by including the value of $r$ in the sketch using $\log n$
bits, which is useful to keep in mind for some of the lower bounds below. However, our main goal is
to determine when \emph{constant-size} sketches are possible, and for this goal distance-invariance
does not make any qualitative difference, as shown in the following simple proposition:

\begin{proposition}
\label{prop:r-invariant equivalence}
If $\F$ admits a constant-size $\alpha$-ADT sketch, then $\F$ admits a constant-size
distance-invariant $\alpha$-ADT sketch.
\end{proposition}
\begin{proof}
Let $s$ be the size of the $\alpha$-approximate distance sketch. Then it holds that for every
$r \in \bN$, the function $D_r$ has domain $\zo^s \times \zo^s$. There are at most $2^{2^{2s}}$
functions $\zo^s \times \zo^s \to \zo$; therefore there are at most $2^{2^{2s}}$ distinct functions
$D_r : \zo^s \times \zo^s \to \zo$. We obtain a distance-invariant $\alpha$-approximate distance
sketch as follows. For each $G \in \F$, each $r \in \bN$, and each $u \in V(G)$, we sample $\sk_r$
from the distribution defined by the constant-size sketch with decoder $D_r$, and we construct
$\sk'(u)$ by concatenating at most $2^{2s}$ bits to specify the function $D_r$. We then define the
decoder $D : \zo^{2^{2s}+s} \times \zo^{2^{2s}+s} \to \zo$ on inputs $\sk'(u), \sk'(v)$ as
$D_r(\sk(u), \sk(v))$, where $D_r$ is the function specified by $\sk'(u)$.
\end{proof}

\noindent
We will show that for monotone classes, ADT sketching implies small-distance
sketching (and therefore adjacency sketching). Any class (monotone or not) that is adjacency
sketchable contains at most $2^{O(n \log n)}$ graphs on $n$ vertices \cite{KNR92,HWZ21}, so any
monotone ADT-sketchable class must satisfy this condition also. One may
wonder if these conditions hold for non-monotone ADT-sketchable classes. The next simple example
shows that this is not so.
\begin{example}
\label{example:general adt}
Consider the class of graphs obtained by choosing any graph $G$ and adding an arbitrary path, with
one endpoint connected to all vertices of $G$. The set of $n$-vertex graphs in this class contains
all $(n-1)$-vertex graphs as induced subgraphs, so it has more than $2^{O(n \log n)}$ graphs and is
not adjacency sketchable. But it is $2$-ADT sketchable: every pair of vertices in $G$ have distance
at most 2 and are equidistant to all other vertices, so we may essentially reduce the problem to a
single path.  Here we have included a path instead of a single vertex so that the graph class has
unbounded diameter.
\end{example}

\subsection{Lower Bound for the Class of All Graphs}

We will write $\fG$ for the class of all graphs, and for an integer
$N$, we will write $\fG_N$ for the class of
all graphs with vertex set $[N]$. It follows from Theorem 3.4 in \cite{TZ05} that
for any $r$,  $\fG$ admits an $(r,\alpha r)$-distance labelling scheme with labels of size
$O(n^{2/\alpha}\log^{2-2/\alpha} n)$.  In \cite{TZ05} a lower bound was also given, but this was for
deterministic approximate distance labels, which must allow an approximate computation of all
distances. We give a stronger result (although the proof is nearly the same) that holds even for the
case where we allow only a $(1,\alpha)$-distance sketch; our bound of $\Omega(n^{1/\alpha})$ has
nearly the same dependence on $\alpha$ as the upper bound.

We will need the following classical result (see Lemma 15.3.2 in~\cite{Mat13} and the references
therein).

\begin{lemma}[\cite{Mat13}]\label{lem:matou}
For any $\ell\ge 2$ and $n\ge 2$, there is an $n$-vertex graph with at
least $\tfrac19 n^{1+1/{\ell}}$ edges and without any cycle of
length at most $\ell+1$. 
\end{lemma}

The proof of the following lower bound is inspired by a seminal proof of Matou\v{s}ek on
non-embeddability of graph metrics in Euclidean space~\cite{Mat96} (see also Proposition 5.1
in~\cite{TZ05} for a closer application on approximate distance oracles). 

\begin{thm}\label{thm:lbdls}
For any $\alpha\ge 2$ and $n\ge 2$, there exists a class
$\F$ of $n$-vertex graphs such that any
 distance-$(1,\alpha)$ labelling scheme for
$\F$ requires labels of size at least $\tfrac19n^{1/\alpha}$. 
\end{thm}

\begin{proof}
For $\alpha\ge 2$ and $n\ge 2$, let $G$ be an $n$-vertex graph with $m\ge \tfrac19 n^{1+\tfrac1{\alpha}}$ edges and without any cycle of
length at most $\alpha+1$ (given by Lemma~\ref{lem:matou}).  Consider
a (deterministic) distance-$(1,\alpha)$ distance labelling scheme for
the class of all spanning subgraphs of $G$. Let
$H$ be a subgraph of $G$. Note that for any edge $uv\in E(G)$, $u$ and
$v$ are at distance 1 in $H$ if $uv\in E(H)$, and are at distance
greater than $\alpha$ otherwise (since $G$ has no cycle of length at most
$\alpha+1$). It follows that given the labels of $u$ and $v$ in $H$, the
decoder outputs 1 if $uv\in E(G)$ and 0 otherwise. Consequently, for
any two distinct subgraphs $H_1,H_2$ of $G$, the sequences of labels
of the vertices $v_1,v_2,\ldots,v_n$ of $G$ in $H_1$ and $H_2$ are
distinct. As there are $2^{m}$ such subgraphs, some subgraph $H$ of
$G$ is such that the sequence of labels of $v_1,v_2,\ldots,v_n$ in $H$
takes at least $m\ge \tfrac19 n^{1+\tfrac1{\alpha}}$ bits, and thus some
vertex of $H$ has a label of size at least $\tfrac19
n^{1/{\alpha}}$, as desired.
\end{proof}

\noindent
Due to Proposition~\ref{prop:derandomization}, we obtain the following immediate corollary.

\begin{corollary}\label{cor:lbdlsrando}
For any $\alpha\ge 2$ and $n\ge 2$, there exists a class
$\F$ of $n$-vertex graphs such that any
 distance-$(1,\alpha)$  sketch for
$\F$ requires labels of size $\Omega(n^{1/\alpha}/\log n)$. 
\end{corollary}

\noindent
Note that in Theorem~\ref{thm:lbdls} and Corollary~\ref{cor:lbdlsrando}, we do not assume that the
distance labelling scheme under consideration is distance-invariant (indeed, we only use the case
$r=1$ to obtain the lower bound). 

\subsubsection{Lower Bound for Bounded-Degree Graphs}

We now prove that a monotone class may have bounded expansion but still have a lower bound of
$n^{\Omega(1/\alpha)}$ on the $\alpha$-ADT sketch size. This bound
holds for the class of graphs of maximum
degree 3, which has expansion exponential in $r$ \cite{NO12}.

Write $\F_{n,3}$ for the class of all $n$-vertex graphs of maximum degree at most 3. We will need
the following construction: Given an $N$-vertex graph $G$ and an
integer $\ell\ge 2\lceil \log N\rceil+1$, let
$G[\ell]$ be any graph obtained from $G$ as follows: each vertex $v$ of $G$ is associated with a
rooted balanced binary tree $T_v$ in $G[\ell]$, whose leaves are indexed by the neighbors of $v$ in $G$
(the trees $T_v$ are balanced, so they have depth at most $\lceil \log
N\rceil $). Then $G[\ell]$ consists in the
disjoint union of all trees $T_v$, for $v\in V(G)$, together with paths connecting the leaf of $T_v$
indexed by $u$ to the leaf of $T_u$ indexed by $v$, for any edge $uv$ of $G$. The length of the path
connecting these two leaves is such that the distance in $G[\ell]$ between  the root of $T_v$ and the root
of $T_u$ is precisely $\ell$.

\begin{thm}\label{thm:cubic}
Assume that there is a real $\alpha$ such that
the class $\F_{n,3}$ has a distance-invariant $\alpha$-ADT sketch of size $s(n)$. Then for any
$\eps>0$, we have $s(n)=\Omega\left(n^{\frac{1}{4\alpha}-\eps}\right)$.
\end{thm}
\begin{proof}
Recall that $\fG$ is the class of all graphs, and $\fG_N$ is the class of all graphs on vertex set
$[N]$.  For a graph $G \in \fG_N $, consider the graph $H:=G[\ell]$ as defined above, with
$\ell=\lceil  4 \log N\rceil$. We denote the root of each tree $T_v$ in $H$ by $r_v$ (see the
paragraph above for the definition of $T_v$). Observe that for any $u,v\in V(G)$,
\[
\tfrac{\ell}2 \cdot \dist_G(u,v) \le (\ell- 2 \log N)\dist_G(u,v) \le \dist_H(r_u,r_v)\le \ell \cdot \dist_G(u,v).
\]
Note that $H\in \F_{n,3}$, with $n\le N\cdot 2N+ {N\choose 2}\cdot 5 \log N\le 8 N^2 \log N$, for
sufficiently large $N$.  We construct a distance-$(1,2\alpha)$ sketch for $\fG_N$ as follows. Let
$D$ be the decoder for the $\alpha$-ADT sketch for $\F_{n,3}$.  Given $G \in \fG_{N}$, the encoder
computes a graph $H$ as above. Since $H \in \F_{n,3}$, for any $r \in \bN$ there is a probability
distribution over functions $\sk_r : V(H) \to \zo^{s(n)}$ such that for all $u,v \in V(H)$:
\begin{align*}
  \dist_H(u,v) \leq r &\implies \Pr[ D(\sk_r(u),\sk_r(v)) = 1 ] \geq 2/3 \,,\\
  \dist_H(u,v) > \alpha r &\implies \Pr[ D(\sk_r(u),\sk_r(v)) = 0 ] \geq 2/3 \,.
\end{align*}
To obtain a sketch for $G$, draw $\sk_\ell : V(H) \to \zo^{s(n)}$ from the appropriate distribution,
and assign to each vertex $v \in V(G)$ the value $\sk'(v) = \sk_\ell(r_v)$. We establish the correctness
of this sketch as follows.  Let $u,v \in V(G)$ and suppose that $\dist_G(u,v) \leq 1$. Then
$\dist_H(r_u,r_v) \leq \ell$, so we have
\[
  \Pr[ D(\sk_\ell(r_u), \sk_\ell(r_v)) = 1 ] \geq 2/3 \,.
\]
Now suppose $\dist_G(u,v) > 2 \alpha$. Then $\dist_H(r_u,r_v) > \alpha \ell$, so we have
\[
  \Pr[ D(\sk_\ell(r_u), \sk_\ell(r_v)) = 0 ] \geq 2/3 \,.
\]
We therefore have a $(1,2\alpha)$-distance sketch for $\fG_N$ of size $s(n)$. Assume for
contradiction that $s(n) = O\left(n^{\tfrac{1}{4\alpha}-\eps}\right)$ for some $\eps >0 $.
By Corollary~\ref{cor:lbdlsrando}, it must be that any distance-$(1,2\alpha)$ sketch for $\G_N$ has
size $\Omega(N^{1/2\alpha}/\log N)$. Therefore we must have $s(n) = \Omega(N^{1/2\alpha}/\log N)$,
so $N^{1/2\alpha}/\log N = O\left(n^{\tfrac{1}{4\alpha}-\eps}\right)$. But
$n \leq 8 N^2 \log N$ for sufficiently large $N$, so
\[
  \frac{N^{1/2\alpha}}{\log N} = O\left(N^{\tfrac{1}{2\alpha} - 2\eps}
\log^{\tfrac{1}{4\alpha}-\eps} N \right) \,,
\]
which is a contradiction.
\end{proof}

\subsection{ADT Sketching Implies Bounded Expansion}

We now prove that if a monotone class $\mathcal{F}$ is ADT sketchable, then $\mathcal{F}$ has
bounded expansion. This is an extension of a similar (unpublished) result for classes of bounded
Assouad-Nagata dimension.

\begin{thm}\label{thm:weakwood}
Let $\F$ be any monotone class of graphs that is $\alpha$-ADT sketchable, for some $\alpha > 1$.
Then $\F$ has bounded expansion.
\end{thm}

\begin{proof}
Let $\mathcal F$ have unbounded expansion, and suppose for the sake of contradiction that it admits
an $\alpha$-approximate distance sketch of constant size $s$.  By Proposition~\ref{prop:r-invariant
equivalence}, we may assume that the sketch is distance-invariant.  Write $D : \zo^s \times \zo^s
\to \zo$ for the decoder, and for every $G \in \mathcal F$ and integer $r$, write $\sk_{G,r} : V(G)
\to \zo^s$ for the associated (random) sketch.

Since $\mathcal F$ has unbounded expansion, by Corollary~\ref{cor:no}, there exists an integer $p
\geq 0$ such that for any integer $t$, $\mathcal F$ contains a $p$-subdivision of a graph of minimum
degree at least $t$. Then, by a recent result of Liu and Montgomery \cite{LM20}, for any integer $t$
there is an integer $k_t \geq 0$ such that $\mathcal F$ contains a $k_t$-subdivision of the complete
graph $K_t$.

Recall that $\fG$ is the class of all graphs. We will design an $\alpha$-approximate distance sketch
for $\fG$. For any $t \in \bN$, let $G \in \fG_t$. Then for $k_t$ defined above,
$\mathcal F$ contains a $k_t$-subdivision of the complete graph $K_t$. Since $\mathcal F$ is
monotone, it also contains the $k_t$-subdivision $G^{(k_t)}$ of $G$. Now observe that, for any $u, v
\in V(G) \subseteq V(G^{(k_t)})$ and integer $r$, we have
\begin{align*}
\dist_G(u,v) \leq r &\implies \dist_{G^{(k_t)}}(u, v) \leq (k_t+1)r \\
\dist_G(u,v) > \alpha r &\implies \dist_{G^{(k_t)}}(u, v) > \alpha (k_t+1)r \,.
\end{align*}
Therefore, with probability at least $2/3$ over the choice of $\sk_{G^{(k_t)}, (k_t+1)r}$, we have
\[
  D(\sk_{G^{(k_t)}, (k_t+1)r}(u), \sk_{G^{(k_t)}, (k_t+1)r}(v)) = \begin{cases}
    1 &\text{ if } \dist_G(u,v) \leq r \\
    0 &\text{ if } \dist_G(u,v) > \alpha r \,,
  \end{cases}
\]
as desired; so $\fG$ admits a distance-invariant $\alpha$-approximate distance sketch of
constant size $s$. But this contradicts Theorem~\ref{thm:lbdls}.
\end{proof}

\noindent
A natural question is whether this can be proved directly, without using the theory of sparsity and
the fairly involved result of Liu and Montgomery \cite{LM20}.

\begin{rem}
We note that the proof of Theorem~\ref{thm:weakwood} gives more than
just
bounded expansion. By using \cite[Theorem 11]{Dvo08} and the recent
quantitative improvements \cite{FHL+22,LTW+22} of the result of Liu and
Montgomery \cite{LM20}, it can be proved that any monotone class with a
distance-invariant 
$\alpha$-ADT sketch of size $s$ has expansion bounded by $f(r)=
O((4rs^{2\alpha})^{(r+1)^2})$. We believe that this bound is far from optimal.
\end{rem}

\subsection{Lower Bound: Grids}

The two-dimensional grid with crosses (also known as the \emph{King's graph}) is a standard example of a class with bounded expansion
\cite{NOW12}.  Let $\mathcal{G}_d$ denote the class of all finite subgraphs of the $d$-dimensional
grid with all diagonals (the strong product of $d$ paths).  It is was proved in \cite{Dvo21} that
for even $d$, the expansion of $\mathcal{G}_d$ is $r\mapsto\Theta(r^{d/2})$, and in particular the
expansion of $\mathcal{G}_2$ is $r\mapsto\Theta(r)$.

\begin{lemma}\label{lem:grid}
For any graph $G$ of maximum degree 3, there exists a constant $k=k_G$ such
that $\mathcal{G}_2$ contains the $k$-subdivision of $G$.
\end{lemma}

\begin{proof}[Proof sketch.]
Observe first that  $\mathcal{G}_2$ contains arbitrarily large
complete graphs as minors (see Figure~\ref{fig:grid}, left, where the
clique minor is obtained by contracting each thick path into a single
vertex). It follows that for any graph $G$, some graph of
$\mathcal{G}_2$ contains $G$ as a minor. Note that for a graph $G$ of maximum
degree at most 3, a graph $H$ contains $G$ as a minor if and only if
$H$ contains a subdivision of $G$. It follows that for any  graph $G$
of maximum degree 3, some sufficiently large grid $H$ of $\mathcal{G}_2$
contains a subdivision of $G$ as a subgraph. For a vertex $u$ of $G$,
we denote by $\pi(u)$ its image in $H$ (in a copy of some subdivision
of $G$ in $H$), and for any edge $uv$ in $G$, we denote the
corresponding path of the subdivision of $G$ in $H$ between $\pi(u)$
and $\pi(v)$ by $P_{uv}$.

\begin{figure}[htb]
 \centering
 \includegraphics[scale=1]{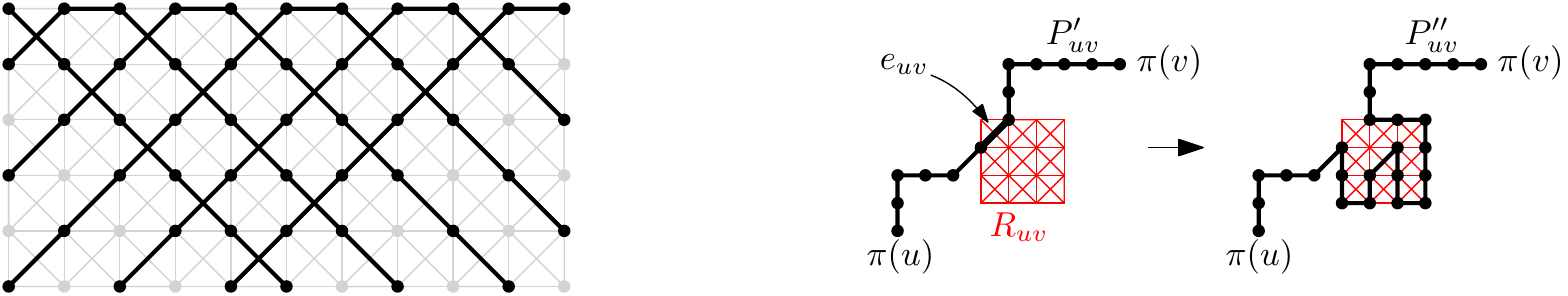}
 \caption{A large clique minor in a 2-dimensional grid with diagonals (left), and
 a way to increase the lengths of the path $P_{uv}'$ using the private
 areas $R_{uv}$ (right).}
 \label{fig:grid}
\end{figure}

For an odd integer $\lambda\ge 1$, a \emph{$\lambda$-refinement} of the grid $H$ is obtained by replacing
each $(1\times 1)$-cell of the grid $H$ by a $(\lambda\times
\lambda)$-grid. Note that if $H$ contains a subdivision of $G$ as
above, then any $\lambda$-refinement of $H$ contains a subdivision of
$G$, where  each path $P_{uv}$ is replaced by a path
$P_{uv}'$ of length $\lambda |P_{uv}|$. Note that after performing a
$\lambda$-refinement and adding small perturbations to some of the paths
$P_{uv}'$, we can
assume in addition that there is an integer $s\ge \lambda/2$ such that
for any edge $uv$ of $G$, $H$
contains a $(s\times s)$-subgrid $R_{uv}$ that intersects the
subdivision of $G$ in $H$ in a single edge $e_{uv}$, which is included
in $P_{uv}'$ (see Figure~\ref{fig:grid}, right).

Using the subgrid $R_{uv}$, we can replace $e_{uv}$ by a path of any
length between 1 and $(\tfrac{\lambda}2-1)^2\ge \lambda^2/16$ (assuming
$\lambda\ge 4$) between its endpoints, turning $P_{uv}'$
into a new path $P_{uv}''$ of length $|P_{uv}'|+\ell=\lambda|P_{uv}|+\ell$ (for any possible
value of $0\le \ell\le \lambda^2/16$), in such a way that the vertices
$\pi(u)$, $u\in V(G)$, and the paths $P_{uv}''$, $uv\in E(H)$, still
form a subdivision of $G$ in $H$.

Note that a path $P_{uv}$ of maximum length can replaced by a path
$P_{uv}''$ of length $\lambda |P_{uv}|$ after the $\lambda$-refinement,
while a path $P_{xy}$ of minimum length can be replaced by a path
$P_{xy}''$ of any length between  $\lambda |P_{xy}|$ and $\lambda
|P_{xy}|+\lambda^2/16\ge \lambda |P_{uv}|$, where the inequality holds
whenever $\lambda\ge 16|P_{uv}|-16|P_{xy}|$. It follows that by
taking $\lambda$ sufficiently large, we obtain a subdivision of $G$ in
$H$ where all edges of $G$ correspond to paths of the same length in
$H$, as desired.
\end{proof}

\noindent
Using a similar argument to the above, we obtain a similar result for subgraphs of 3-dimensional
grids. Write $\P^3$ for the class of subgraphs of the 3-dimensional grid; \ie the class of finite
subgraphs of the Cartesian product $P^3$, where $P$ is the infinite path. We omit this proof due to
its similarity to the one above.
\begin{lemma}\label{lem:grid3}
For any graph $G$ of maximum degree 3, there exists a constant $k=k_G$ such that
$\P^3$ contains the $k$-subdivision of $G$.
\end{lemma}

\noindent
We easily deduce the following simple corollary.

\begin{corollary}
The classes $\mathcal{G}_2$ and $\P^3$ are not ADT sketchable.
\end{corollary}
\begin{proof}
The proof is similar to that of Theorem~\ref{thm:weakwood}.  Suppose for contradiction that, for
some constant $\alpha > 1$, the class admits a constant-size $\alpha$-approximate
distance sketch. Then by Proposition~\ref{prop:r-invariant equivalence}, we can assume that the
sketch is distance-invariant. By Lemma~\ref{lem:grid}, this can be used to design an
$\alpha$-approximate distance sketch for the class of all graphs of maximum degree 3, contradicting
Theorem~\ref{thm:cubic}.
\end{proof}

\subsection{Lower Bound for Classes of Low Expansion}

Now we show that there is no non-constant bound on the expansion that guarantees the existence of
constant-size ADT sketches. We achieve this by constructing classes of graphs of arbitrarily low
non-constant expansion, which cannot admit constant-size $\alpha$-ADT sketches for any constant
$\alpha > 1$.  We start with a simple variant of \cite[Theorem 4.5]{GKRSS18}.

For a function $f:\mathbb{N}\to \mathbb{N}$ such that $f(n)\to\infty$ when $n\to \infty$, we define
$f^{-1}(n):=\max \{k\,|\,f(k)\le n\}$. We recall that for an
$N$-vertex graph $G$ and an integer $\ell\ge 2 \log N+1$, the graph
$G[\ell]$ was defined just before the statement of Theorem \ref{thm:cubic}.

\begin{lemma}\label{lem:GKRSS}
Let $f$ be a function such that $f(n)\to\infty$ when $n\to
\infty$. For any $n$-vertex graph $G$, and any integer $r\ge 0$, every
depth-$r$ minor of  $G[6f(6n^2)+2\log n]$ has average degree at most $\max\{4,  f^{-1}(r)\}$.
\end{lemma}

\begin{proof}
Let $H=G[6f(6n^2)+2\log n]$, and let $H'$ be any depth-$r$ minor of $H$.  Observe that $G[2\log
n+1]$ has at most $2n^2+{n\choose 2}\le 3n^2$ vertices, and thus tree-width at most $3n^2$. The
graph $H$ itself is obtained from $G[2\log n+1]$ by subdividing some edges, an operation that leaves
the tree-width unchanged. It follows that $H$ also has tree-width at most $3n^2$. Since tree-width
is a minor-monotone parameter, $H'$ also has tree-width at most $3n^2$ and is thus
$3n^2$-degenerate. It follows that $H'$ has average degree at most $6n^2$. If $r\ge f(6n^2)$, then
$6n^2\le f^{-1}(r)$ and thus $H'$ has average degree at most $f^{-1}(r)$, as desired. Assume now
that $r\le f(6n^2)$. In this case, since $H'$ is obtained from disjoint trees by connecting their
leaves with paths of length at least $6 f(6n^2)$, it can be checked that $H'$ is 2-degenerate, and
thus has average degree at most 4. 
\end{proof}

\begin{thm}\label{thm:arbexp}
For any function $\rho$ tending to infinity, there exists a monotone class of expansion $r \mapsto
\rho(r)$ that is not ADT sketchable.  Moreover, for any $\eps>0$, there exists a monotone class
$\F$ of expansion $r\mapsto O(r^\eps)$, such that, if $\F$ admits an $\alpha$-ADT sketch of size
$s(n)$, then we must have $s(n) = n^{\Omega(1/\alpha)}$.
\end{thm}

\begin{proof}
Let $\rho : \bN \to \bN$ be a function tending to infinity, so that $\rho^{-1}$ is a non-decreasing
function tending to infinity. We proceed as in the proof of Theorem~\ref{thm:cubic}, setting
$\ell(N) = 6 \rho^{-1}(6N^2) + 2\log N$ (instead of $\ell = \lceil 4 \log N \rceil$). For any
$N$-vertex graph $G$, $G[\ell(N)] \in \F_{n,3}$ with $n \leq 6N^2(\rho^{-1}(6N^2) + 2 \log N)$. By
Lemma~\ref{lem:GKRSS}, any depth-$r$ minor of such a graph $G[\ell(N)]$ has average degree at most
$\max\{4, \rho(r)\}$. It follows that the monotone class $\F$ of all graphs $G[\ell(N)]$ for $G \in
\fG_N$ and their subgraphs has expansion at most $r \mapsto \max\{4,\rho(r)\}$.

By the same argument as in Theorem~\ref{thm:cubic}, if there is a distance-invariant $\alpha$-ADT
sketch for $\F$ of size $s(n)$ (which is a non-decreasing function in $n$), we obtain a
$(1,2\alpha)$-distance sketch for $\fG$ of size $N \mapsto s(n)$. Then, due to
Corollary~\ref{cor:lbdlsrando},
\[
N^{\tfrac{1}{2\alpha}} / \log N = O\left(s(n)\right) = O\left( s \left( 6N^2(\rho^{-1}(6N^2) + 2
\log N) \right)\right) \,.
\]
It is clear that, for any choice of $\rho$, we cannot have $s(n)$ constant, which establishes the
first part of the theorem. To get the second part, let $\eps > 0$ and suppose that we choose
$\rho(r) = r^\eps$ so that $\rho^{-1}(r) = r^{1/\eps}$, and assume for contradiction that
$s(n) = n^{o(1/\alpha)}$. Then
\begin{align*}
N^{\tfrac{1}{2\alpha}} / \log N
  = O\left( s\left( 6N^2((6N^2)^{1/\eps} + 2\log N)\right) \right)
  = O\left( N^{o(1/\alpha)} \right) \,,
\end{align*}
which means we must have $s(n) = n^{\Omega(1/\alpha)}$ as desired.
\end{proof}

\subsection{Upper Bounds}
\label{sec:upperbounds}
The \emph{weak diameter} of a subset $S$ of vertices of a graph $G$ is
the maximum distance in $G$ between two vertices of $S$.
Given a graph $G$, a \emph{$(\sigma, \tau, \Delta)$-sparse cover} is a family
$\mathcal{C}$ of
subsets of $V(G)$ of weak diameter at most $\Delta$, such that (i) for each $u\in V(G)$, there is a set
$C\in \mathcal{C}$ such that $B(u,\tfrac{\Delta}{\sigma})\subseteq C$ (where
$B(u,r)$ denotes the ball of radius $r$ centered in $u$), and (ii)
each vertex $u\in V(G)$ lies in at most $\tau$ sets of $\mathcal{C}$.

We say that a graph $G$ admits a \emph{$(\sigma, \tau)$-sparse cover
  scheme} if for any $\Delta$, it admits a $(\sigma, \tau, \Delta)$-sparse
cover. We say that a graph class $\F$ has a $(\sigma, \tau)$-sparse cover
  scheme if any graph of $\F$ has such a scheme\footnote{It
    is usually assumed that in addition, such schemes can be computed
    efficiently, that is in time polynomial in the size of the
    graph.}). Classes of graphs with $(\sigma, \tau)$-sparse cover
  schemes are also known as classes of \emph{Assouad-Nagata dimension} at
most $\tau-1$ in metric geometry~\cite{As82} (see also~\cite{LS05}).

 We recall that
disjunctive labelling schemes have been defined in Subsection
\ref{sec:dis}, and that by Proposition \ref{prop:eqbased} they can be
turned into constant-size sketches with one-sided error. The following result is a simple consequence of the definition of
 sparse covers. Note
that here the size of the labels is independent of $r$, in contrast with the setting of
Theorem~\ref{thm:distr1} and its corollaries. 

\begin{thm}\label{thm:SCdist}
If a graph class $\F$ has  a $(\sigma, \tau)$-sparse cover
scheme, then $\mathcal{F}$ has distance-invariant, disjunctive $\sigma$-ADT
labelling scheme with labels of size  $O(\tau)$.
\end{thm}

\begin{proof}
By the definition of sparse covers, for any $r>0$, there is a family
$\mathcal{C}$ of
subsets of $V(G)$ of weak diameter at most $\sigma r$, such that (i) for each $u\in V(G)$, there is a set
$C\in \mathcal{C}$ such that $B(u,r)\subseteq C$, and (ii)
each vertex $u\in V(G)$ lies in at most $\tau$ sets of $\mathcal{C}$.

We may now define a disjunctive $\sigma$-ADT labelling scheme. Assign
each set $C \in \mathcal{C}$ a unique number in $\bN$, and for each
vertex $x$, let $S(x)$ be the set of names of the (at most
$\tau$) sets $C$ containing $x$.  For each vertex
$x \in V(G)$, the equality code $\vec q(x)$ contains the names $S(x)$, and we assign $x$ the label $(- \mid \vec q(x))$. On inputs $(- \mid \vec q(x))$ and $(- \mid \vec q(y))$, the decoder outputs 1 if and
only $S(x)\cap S(y)\ne \emptyset$ (which can be checked using the equality codes).

Suppose $\dist_G(x,y) \leq r$. Since there is a set $C\in \mathcal{C}$
such that $B(x,r)\subseteq C$, we also have $y\in C$, and thus
$S(x)\cap S(y)\ne \emptyset$. Now suppose that $\dist_G(x,y) > \sigma
r $. Since each set $C\in \mathcal{C}$ has weak diameter at most
$\sigma r$, there is no set $C\in \mathcal{C}$ containing both $x$ and
$y$ and thus $S(x)\cap S(y)=\emptyset$.
\end{proof}

Using Proposition \ref{prop:eqbased}  together with results of~\cite{Fil20} on sparse covers (based
on~\cite{KPR93,FT03}), we deduce the following immediate corollary.

\begin{corollary}\label{cor:minor}
For any $t\ge 4$, the class of  $K_t$-minor free graphs has
a distance-invariant  $O(2^t)$-ADT sketch of size 
$O(t^2\log t)$.
\end{corollary}

\subsubsection{Padded Decompositions}
We will see now how to obtain improved sketches using \emph{padded decompositions}. These improved
sketches have some disadvantages: they have two-sided error whereas the other sketches in this paper
have one-sided error, and they are not equality-based.  For a graph $G$, a probability distribution
$\mathcal{P}$ over partitions of $V(G)$ is said to be \emph{$(\beta,\delta,\Delta)$-padded} if
\begin{itemize}
\item for each partition $P$ in the support of $\mathcal{P}$, each set of $P$ has
  weak diameter at most $\Delta$, and
\item for any $u\in V(G)$ and any $0\le \gamma\le \delta$,
  the ball $B(u,\gamma \Delta)$ is included in some set of a random
  partition from $\mathcal{P}$ with probability at least $2^{-\beta\gamma}$.
\end{itemize} 
\noindent
We say that a graph $G$ admits a \emph{$(\beta, \delta)$-padded decomposition scheme} if for any
$\Delta$, it admits a $(\beta, \delta, \Delta)$-padded distribution over partitions of $V(G)$. We say that a
graph class $\F$ has a $(\beta, \delta)$-padded decomposition scheme if any graph of $\F$ has such a
scheme. We note that padded decomposition schemes are incomparable with the
sparse cover schemes introduced above.

\begin{thm}\label{thm:PDSdist}
If a graph class $\F$ has  a $(\beta, \delta)$-padded decomposition scheme, then $\F$ admits a
distance-invariant $\alpha$-ADT sketch of size 2, where $\alpha=\max(\tfrac1{\delta},
\tfrac{\beta}{\log(3/2)})$.
\end{thm}

\begin{proof}
Fix some $r>0$, and some graph $G\in \F$. Let $\gamma=
\min(\delta,\tfrac{1}{\beta}\log(3/2))$. Note that $0\le \gamma\le
\delta$ and $2^{-\beta \gamma}\ge \tfrac23$. Let $\Delta=r/\gamma$,
and let $\mathcal{P}$ be a $(\beta,\delta,\Delta)$-padded distribution over partitions of
$V(G)$. Let $\mathbf{P}$ be a partition of $V(G)$ drawn according to
the distribution $\mathcal{P}$. Then each set of $\mathbf{P}$ has weak
diameter at most $\Delta$. Assign a random identifier
$\text{id}(S)$ to each set $S\in \mathbf{P}$, drawn uniformly at
random from the set $\{1,2,3\}$ (each $\text{id}(S)$ requires only 2 bits). The label of each vertex $u\in V(G)$ simply consists of the identifier
$\text{id}(S)$ of the unique set $S\in \mathbf{P}$ such that $u\in
S$. Given the labels of $u$ and $v$, the decoder outputs 1 if and
only the labels are equal. Note that the decoder is clearly distance-invariant.

Assume first that $d(u,v)\le r$, then since $v$ is in
the ball of radius $r=\gamma \Delta$ centered in $u$, it follows that $u$
and $v$ are in the same set $S\in \mathbf{P}$ with probability at
least $2^{-\beta \gamma}\ge \tfrac23$. If $u$
and $v$ are in the same set $S\in \mathbf{P}$, then their labels are
equal with probability 1. It follows that if $d(u,v)\le r$, the
decoder outputs 1
with probability at least $2/3$, as desired.

Assume now that $d(u,v)> r/\gamma$. Since each set in $\mathbf{P}$ has weak
diameter at most $\Delta=r/\gamma$, it follows that $u$ and $v$ are in different sets of
$\mathbf{P}$ with probability 1. As each set $S\in \mathbf{P}$
is assigned a random element from $\{1,2,3\}$, $u$ and $v$ have the
same label with probability $\tfrac13$. It follows that if $d(u,v)>
r/\gamma$, the decoder outputs 0 with probability
$1-\tfrac13=\tfrac23$, as desired.
\end{proof}

Although this sketch \emph{does} use randomization for an equality check, it also uses randomization
to construct the padded decomposition, and so it is not equality-based.

It was proved in \cite{AGGNT19} that for $t\ge 4$, the class of  $K_t$-minor free graphs admits a
$(320t, \tfrac1{160})$-padded decomposition scheme. It was also proved in~\cite{LS10} (see also
\cite{AGGNT19}) that for any $g\ge 0$, the class of graphs embeddable on a surface of Euler genus
$g$ admits a $(O(\log g), \Omega(1))$-padded decomposition scheme. We obtain the following two
corollaries, and again emphasize that these sketches have two-sided error.

\begin{corollary}\label{cor:minor2}
For any $t\ge 4$, the class of  $K_t$-minor free graphs has a
distance-invariant $O(t)$-ADT sketch with labels of at most 2 bits.
\end{corollary}

\begin{corollary}\label{cor:genus}
For any $g\ge 0$, the class of  graphs embeddable on a surface of
Euler genus $g$ has  a distance-invariant $O(\log g)$-ADT sketch with labels of at most 2 bits.
\end{corollary}

We remark that lower bounds for sketching therefore imply lower bounds
on sparse covers and padded
decompositions. For example, the communication complexity lower bound of
\cite{AK08} implies that bounded-degree expanders do not admit padded covers or sparse covers, which
is something that we were unable to prove directly.

\begin{acknowledgement}
We thank Gwena\"el Joret for many helpful discussions, and for allowing us to include his proof of
Theorem~\ref{thm:distr2}. We thank Viktor Zamaraev for leading us to Corollary~\ref{cor:hypercube
counterexample} and carefully proofreading our manuscript. We thank Alexandr Andoni for a helpful discussion and for sharing with us the
manuscript \cite{AK08}. We thank Renato Ferreira Pinto Jr. and Sebastian Wild for comments on the
presentation of this article. Finally, we thank the anonymous reviewers of the
journal and 
conference \cite{EHK22} versions of this paper for the many insightful
comments and suggestions. We especially thank a reviewer for suggesting the simplified and improved
version of Lemma~\ref{lemma:min degree}.
\end{acknowledgement}

\bibliographystyle{halpha}
\bibliography{references.bib}

\end{document}